\documentclass[authoryear,1p,10.8pt]{elsarticle}
\usepackage{algorithmicx}
\usepackage{algorithm}

\usepackage{algpseudocode}

\usepackage{graphicx} 

\usepackage{amsmath, amssymb, amsthm,  bbm} 

\newtheorem{theorem}{Theorem}[section]
\newtheorem{lemma}{Lemma}[section]
\newtheorem{corollary}{Corollary}[section]
\newtheorem{remark}{Remark}[section]
\newtheorem{definition}{Definition}[section]

\newtheorem{example}{Example}[section]
\newtheorem{problem}{Problem}[section]
\usepackage{verbatim}
\usepackage[bb=dsserif]{mathalpha}
\usepackage{bm}

\usepackage{tikz}
\usepackage{color,hyperref}
\definecolor{darkblue}{rgb}{0.0,0.0,0.5}
\hypersetup{colorlinks,breaklinks,linkcolor=darkblue,urlcolor=darkblue,anchorcolor=darkblue,citecolor=darkblue}

\usepackage{makecell}
\usepackage{tikz}
\usetikzlibrary{matrix,positioning}
\tikzset{bullet/.style={circle,fill,inner sep=2pt}}
\usepackage{siunitx}
\usepackage{subcaption}

\usepackage{natbib}
\usepackage{booktabs}
\usepackage{pgf-pie}
\usepackage{pgfplots}
\pgfplotsset{compat=newest}

\usetikzlibrary{decorations.markings}

\tikzset{
tangent/.style={ 
	decoration={
		markings,
		mark=
		at position #1
		with
		{
			\coordinate (tangent point-\pgfkeysvalueof{/pgf/decoration/mark info/sequence number}) at (0pt,0pt);
			\coordinate (tangent unit vector-\pgfkeysvalueof{/pgf/decoration/mark info/sequence number}) at (1,0pt);
			\coordinate (tangent orthogonal unit vector-\pgfkeysvalueof{/pgf/decoration/mark info/sequence number}) at (0pt,1);
			\draw[blue] (0pt,0pt) coordinate (Tang) -- (20pt,0pt);
		}
	},
	postaction=decorate
},
use tangent/.style={
	shift=(tangent point-#1),
	x=(tangent unit vector-#1),
	y=(tangent orthogonal unit vector-#1)
},
use tangent/.default=1
}
\makeatletter

\DeclareRobustCommand{\rvdots}{%
\vbox{
	\baselineskip4\p@\lineskiplimit\z@
	\kern-\p@
	\hbox{.}\hbox{.}\hbox{.}
}}

\usepackage[appendix = inline]{apxproof}
\newtheoremrep{theorem}{Theorem}[section]
\newtheoremrep{lemma}{Lemma}[section]

\let\OLDthebibliography\thebibliography
\renewcommand\thebibliography[1]{
\OLDthebibliography{#1}
\setlength{\parskip}{0 pt}
\setlength{\itemsep}{4pt plus 0.3ex}
}

\journal{European Journal of Operation Research}

\begin{document}

\begin{frontmatter}
	
	
	
	\title{On Accelerating Large-Scale\\ Robust Portfolio Optimization}
	
	
	\author[inst1]{Chung-Han Hsieh} 
	
	\affiliation[inst1]{organization={Department of Quantitative Finance, National Tsing Hua University},
		addressline={\\ No. 101, Section~2, Kuang-Fu Road}, 
		city={Hsinchu},
		postcode={30013}, 
		country={Taiwan}}
	
 \author[inst1]{Jie-Ling Lu}
	
	
	\begin{abstract}
Solving large-scale robust portfolio optimization problems is challenging due to the high computational demands associated with an increasing number of assets, the amount of data considered, and market uncertainty. 
To address this issue, we propose an extended supporting hyperplane approximation approach for efficiently solving a class of distributionally robust portfolio problems for a general class of \emph{additively separable utility} functions and polyhedral ambiguity distribution set, applied to a large-scale set of assets. 
Our technique is validated using a large-scale portfolio of the S\&P 500 index constituents, demonstrating robust out-of-sample trading performance.  
More importantly, our empirical studies show that this approach significantly reduces computational time compared to traditional concave Expected Log-Growth (ELG) optimization, with running times decreasing from several thousand seconds to just a few. This method provides a scalable and practical solution to large-scale robust portfolio optimization, addressing both theoretical and practical~challenges.
	\end{abstract}
	
	
%
%
%
%
%
%
	
	
\end{frontmatter}


\section{Introduction } \label{section: Introduction}

Solving large-scale robust portfolio optimization problems presents significant computational challenges due to the increasing number of assets, the amount of data considered, and market uncertainty. 
Traditional approaches, such as the mean-variance (MV) models proposed by  \cite{markowitz1952}, are single-period in nature, and assume full availability of the return distribution.
Standard multi-period approaches, such as dynamic programming and stochastic control frameworks, often face scalability issues and require strong assumptions about return distributions and investor preferences.

In response to these challenges, this paper introduces a novel method that significantly enhances the computational efficiency of solving a class of robust portfolio optimization problems for a general class of additively separatable utility functions, polyhedral ambiguity sets, and proportional transaction costs.

We propose an extended supporting hyperplane approximation incorporating turnover transaction costs, significantly enhancing the computational efficiency of solving robust portfolio optimization problems for a general class of \emph{additively separable} utility functions and polyhedral ambiguity distribution sets.
Specifically, our key contributions are as follows:

\emph{Generalized Supporting Hyperplane Approximation:}
We develop an extended supporting hyperplane approximation method tailored to robust portfolio optimization. This method efficiently handles additively separable utility functions and polyhedral uncertainty sets while incorporating proportional transaction costs. Unlike existing methods, our approach significantly reduces the computational burden and scales well with the number of assets.

\emph{Approximation Error Analysis:}
We provide a comprehensive error analysis of our approximation method. The total approximation error is decomposed into errors arising from returns and transaction costs, enhancing the tractability and reliability of our approach. This analysis is crucial for understanding the trade-offs involved and for ensuring robust performance in practical applications.

\emph{Empirical Validation:}
Our empirical studies validate the robustness and efficiency of the proposed method using a large-scale portfolio of S\&P 500 index constituents. We demonstrate a significant reduction in computational time from several thousand seconds to just a few seconds, compared to the traditional expected logarithmic growth (ELG) optimal portfolio framework. Our results show robust out-of-sample trading performance, highlighting the practical applicability of our approach.

\subsection{Background and Related Work}
The foundation of modern portfolio optimization was established with the well-known mean-variance~(MV) model proposed by \cite{markowitz1952}, which describes the trade-off between risks and returns. 
Researchers have since explored various extensions, including different risk measures, such as value at risk (VaR), e.g., \cite{duffie1997overview, jorion2007value} and conditional value at risk~(CVaR), e.g., \cite{rockafellar2000optimization} as well as other robust statistical approaches to mitigate the sensitivity to parameter inputs, e.g., \cite{black1992global, feng2016signal}. 
Despite its significance, the MV model is mainly single-period in nature and is sensitive to parameter inputs, making it error-prone, as discussed in \cite{michaud1989markowitz} and \cite{fabozzi2007robust}. A comprehensive review of the topic can be found in \cite{steinbach2001markowitz}.

In contrast to single-period models, the \emph{Kelly criterion} proposed by \cite{kelly1956new} maximizes the expected logarithm growth (ELG) of wealth in a repeated betting game. 
The ELG model has desirable theoretical properties, such as  \emph{myopic} optimization if returns are known to be independent and identically distributed, see \cite{Cover_Thomas_2012, MacLean_Thorp_Ziemba_2011book}. 
Extensions of the ELG model include log-mean-variance criteria in portfolio choice problems, see  \cite{luenberger1993preference} and asymptotic optimality in a rebalancing frequency-dependent setting, see \cite{hsieh2023asymptotic}.
It is known that this Kelly-type investment can be cast into a general expected utility theory framework with logarithmic utility; see \cite{luenberger2013investment}.
However, in practice, the actual return distribution is unavailable to the investor.  Hence, constructing an empirical distribution solely from historical returns may lead to over-fitting, resulting in poor out-of-sample trading performance.

\subsubsection{Distributional Robust Portfolio Optimization}
Given the true return distributions are often unknown, leading to \emph{ambiguity} for investors, a broad literature focuses on distributionally robust optimization (DRO) approaches. 
In DRO problems, an ambiguity set defines a family of return distributions consistent with some known~information. 
Examples include maximizing the growth rate of the worst-case VaR, see \cite{rujeerapaiboon2016robust}, extensions to autocorrelated return distributions \cite{choi2016multi}, and ambiguity regions involving means and covariances of return vectors 
\cite{delage2010distributionally}. General results connecting DRO and optimal transport theory are studied in \cite{blanchet2019quantifying} and the data-driven approach of the DRO problem is studied in \cite{mohajerin2018data}.

Furthermore, \cite{blanchet2022distributionally} examined the distributionally robust version of the MV portfolio selection problem with Wasserstein distance. 
Other studies, such as \cite{sun2018distributional}, derive distributionally robust Kelly problems from various ambiguity sets, such as polyhedral, ellipsoidal, and Wasserstein~sets; see also the comprehensive text by~\cite{shapiro2021lectures}. 
Recently, \cite{hsieh2024solving} addressed a DRO version of the Kelly problem with a polyhedral ambiguity set using the supporting hyperplane approximation approach, including practical constraints, and achieving significant computational improvement for a mid-sized portfolio selection problem. Additionally, \cite{li2023wasserstein} solved a Wasserstein-Kelly problem by leveraging a log-return transformation and convex conjugate approach. 

While many prior studies focus on log-utility, such as the Kelly criterion-based literature, without transaction costs, our approach accounts for a general class of additively separable utilities and incorporates market friction arising from turnover transaction costs. This extension greatly broadens the applicability and practical relevance of the method.
Reviews of the recent DRO approach are available in \cite{rahimian2019distributionally}. For specific applications, we refer to \cite{ghahtarani2022robust} the references~therein.

\subsubsection{Large-Scale Considerations}
Beyond distributional robustness, another critical aspect of portfolio optimization is the computational complexity associated with large-scale portfolios.
Early solutions include sparsifying the covariance matrix of asset returns to reduce nonzero elements \cite{perold1984large} and using static mean-absolute deviation portfolio optimization model by \cite{konno1991mean}, and the compact mean-variance-skewness model by \cite{ryoo2007compact}.
Additionally,  \cite{takehara1993interior} proposed the interior point algorithm, demonstrating that the increase in CPU time is almost linear to the problem size~$n$. 
\cite{potaptchik2008large} adopted the mean-variance model with nonlinear transaction costs.

Recent developments focus on algorithmic and stochastic programming approaches, such as the ``value function gradient learning'' algorithm for large-scale multistage stochastic convex programs \cite{lee2023value} and the sample path approach to multistage stochastic linear optimization \cite{bertsimas2023data}. 
However, further exploration of their computational efficiencies and practical scalability in empirical justifications is needed.

Our method addresses these challenges by providing a scalable and robust solution for large-scale portfolio management. The extended supporting hyperplane approximation, combined with turnover transaction costs, enhances computational efficiency and robustness, making it a practical and effective tool for portfolio optimization.

\subsection{Notations} \label{subsection: Notations}

In this paper, we use the following notations: $\mathbb{R}^n$ denotes the $n$-dimensional Euclidean space. The~$\ell_p$-norm of a vector $z \in \mathbb{R}^n$ is denoted by $\|z\|_p = \left(\sum_{i=1}^n |z_i|^p\right)^{1/p}$. A \emph{concave function} $f: \mathbb{R}^n \to \mathbb{R}$ satisfies~$f(\lambda x + (1-\lambda) y) \geq \lambda f(x) + (1-\lambda) f(y)$ for any $x, y \in \mathbb{R}^n$ and $\lambda \in [0,1]$. A function $f$ is \emph{convex} if $-f$ is concave, and it is \emph{strictly concave} if the inequality is strict for $x \neq y$ and $\lambda \in (0,1)$. 
Let $g(a, b)$ be real-valued function defined on $\mathbb{R}^n \times \mathbb{R}^n$, then $g(a, b)$ is \emph{jointly convex} if for all~$ x_1, x_2, y_1, y_2 \in \mathbb{R}^n$ and $\lambda \in [0,1]$,
$
g( \lambda x_1 + (1 - \lambda)x_2, \lambda y_1 + (1 - \lambda)y_2) \leq \lambda g(x_1, y_1) + (1 - \lambda)g(x_2, y_2).
$
Moreover, we say that $g(x, y)$ is \emph{jointly concave} if $-g(x, y)$ is jointly convex.
All random objects are defined in a probability space $(\Omega, \mathcal{F} , \mathbb{P})$ with $\Omega$ being the sample space,~$\mathcal{F}$ being the information set, and $\mathbb{R}$ being the probability measure. 
Notation $\mathbb{E}_p[\cdot]$ represents the expectation operator with respect to a probability distribution $p$. The probability simplex $S_m$ is defined as $S_m := \{p \in \mathbb{R}^m_+ : p^\top \mathbf{1} = 1\}$, where $\mathbb{R}^m_+$ is the set of non-negative $m$-dimensional vectors and $\mathbf{1} \in \mathbb{R}^m$ is a vector of ones. The leverage constant $L \geq 1$ defines the upper bound on the sum of the absolute values of portfolio weights. The turnover rate $TR(t)$ quantifies the changes in portfolio weights $K(\cdot)$ between periods, defined as $TR(t) := | K(t) - K( t-1 ) |^\top \mathbf{1}$. The auxiliary function $f_t(x, c) := U_t( (1 + x)(1 - c) )$ represents a utility function adjusted for returns and transaction costs, where $U_t$ is continuously differentiable, strictly monotonic, and strictly concave.

\section{Preliminaries} \label{section: Preliminaries}

This section provides some preliminaries useful for the problem formulation.
Specifically, consider a financial market with $\mathfrak{N} > 1$ assets. We form a portfolio of $1 \leq n \leq \mathfrak{N}$ assets, to be rebalanced with weights $K(t) \in \mathbb{R}^n$ during rebalancing periods~$t = 1, 2, \dots, T$. For $i = 1, 2, \dots, n$, let $S_{i}(t) > 0$ be the price of Asset~$i$ at period $t$.  The associated per-period return is given by:
$$
X_{i}(t) = \frac{S_{i}(t)-S_{i}(t-1)}{S_{i}(t-1)}
$$
with $\min\{X_{i}(t) : t= 1, 2, \dots, T,\ i = 1, 2, \dots, n\} > -1$. 
The price vector is defined as $S(t) := [S_{1}(t) \; S_{2}(t) \; \cdots  S_{n}(t)]^\top$, and the return vector is defined as~$X(t) := [X_{1}(t) \; X_{2}(t) \; \cdots X_{n}(t)]^\top $.
These returns are treated as a random sample following an unknown distribution.
Henceforth, we assume that the returns are independent and identically distributed (i.i.d.) over time $t$, but with some unknown joint distribution $p$ with $m$ supports. That is,~$\mathbb{P}(X(t) = x^j) = p^j$ for $j = 1, \dots, m$. Henceforth, we take $x_{i, \min} := \min_j x_i^j$ and $x_{i, \max} := \max_j x_i^j$ for $i = 1,\dots, n.$


\begin{remark}[Assumption of Return Model]\rm
Consistent with \cite{mohajerin2018data, li2023wasserstein,hsieh2024solving},
we note that the random returns $X_i$ are treated as independent and identically distributed (i.i.d.) samples from the true but unknown return distribution, which is assumed to lie within an ambiguous distribution set.
This assumption enables the application of robust statistical methods for analysis. Empirical evidence suggests that financial returns over relatively short intervals can be approximated by i.i.d. samples, simplifying our mathematical framework without sacrificing accuracy. 
However, the i.i.d. assumption may not fully capture temporal dependencies and dynamic correlations. To address this, we use an ambiguous distribution set, enhancing robustness by considering a range of potential distributions. This approach mitigates overfitting and safeguards against extreme market conditions, making the model more realistic and applicable to diverse market behaviors.
\end{remark}

\subsection{Account Value Dynamics with Turnover Transaction Costs}

Let $V(t)$ be the account value at period $t$, where the initial account value is $V(0) > 0$. With $K_{i}(t)$ being the weight of the $i$th asset invested at period $t$ for all $i$, we take $
K(t) := [K_{1}(t)\; K_{2}(t)\; \cdots \; K_{n}(t) ]^\top.
$
For transaction costs, we consider a fixed rate charged on the turnover for the $i$th asset. Hence, the \emph{turnover transaction cost} $\mathfrak{C}(t)$ at period $t$ is given by
$$
\mathfrak{C}(t) := |K(t) - K(t-1)|^\top V(t-1)C(t)
$$
where $C(t) := [c_1(t) \; c_2(t)\; \cdots \; c_n(t)]^\top$ is the \emph{cost} vector with $c_i(t) \in [0,1)$ for $i \in \{1, 2, \dots, n\}$ and~$|K(t) - K(t-1)|$ means the componentwise absolute value with the $i$th element being~$|K_i(t) - K_i(t-1)|$.
Hence, the \emph{cost-adjusted account value} at period $t$~satisfies the following stochastic recursive equation: For~$t = 1, 2, \dots, T$,
\begin{align*}
V(t) 
&= ( 1 + K(t)^\top   X(t)) \left( V(t-1) - \mathfrak{C}(t) \right) \\
&= (1 + K(t)^\top   X(t))(1 - |K(t) - K(t-1)|^\top C(t) ) V(t-1)
\end{align*}
with initial account value $V(0) > 0.$

\begin{remark}[Zero Cost Case] \rm
It is readily verified that when there are no transaction costs, i.e.,~$C(t) \equiv \textbf{0} \in \mathbb{R}^n$, the account value dynamics above reduce to
$    
V(t) = (1 + K(t)^\top X(t)) V(t-1).
$
which is consistent with the existing models; see, e.g.,~\cite{li2018transaction, hsieh2023asymptotic}.
If the regulator specifies constant rates, we set $C(t) \equiv C$ with $c_i(t) := c_i$ for all $t$.
\end{remark}

\subsection{Practical Trading Constraints} \label{subsection: Trading Constraints}
This section discusses various trading constraints that will be imposed on the model.

\subsubsection{Leverage, Short Selling, Turnover, and Holding Constraints.}
To impose the leverage constraint, we take $L \geq 1$ as the \emph{leverage constant}. Then, the leverage constraint is given by $\sum_{i=1}^n |K_i(t) V(t)| \leq LV(t)$, which implies 
\begin{align} \label{eq: leverage constraint}
\sum_{i=1}^n |K_i(t)| \leq L.
\end{align}
When $L = 1$, it corresponds to being \emph{cash-financed}; if $L > 1$, it corresponds to \emph{leverage}. 
On the other hand, by allowing \emph{shorting}, the constraint is written as
$
\sum_{i=1}^n |K_i^{+}(t) + K_i^{-}(t)| \leq L,
$
where $K_i^{+}(t) > 0$ and $K_i^{-}(t) < 0$ represent the proportion of longing and shorting in the $i$th asset, respectively. To go long, we require the sum~$K_i^{+}(t) + K_i^{-}(t) > 0$. Similarly, to go short, the sum must satisfy $K_i^{+}(t) + K_i^{-}(t) < 0$.

Additionally, the portfolio turnover may result in large transaction costs, making the rebalancing inefficient. To this end, one may restrict the amount of turnover allowed as a constraint. Typically, we restrict $|K_i(t+1) - K_i(t)| \leq U_i$ or on the whole portfolio $\|K(t+1) - K(t)\|_1 \leq U$ for some constant $U$, where $\|z\|_1$ denotes the $\ell_1$-norm for $z \in \mathbb{R}^n$ with $\|z\|_1 := \sum_{i=1}^n |z_i|.$
Lastly, for the sake of risk management, concentrated holdings can be avoided by constraining the upper bound of the portfolio weight, i.e., for all $i = 1, \dots, n$, and time $t$,
\begin{align} \label{eq: diversification constraint}
|K_i(t)| \leq D_i
\end{align}
for some constant $D_i > 0$. If $D_i := \frac{L}{n}$ for all $i$, this is referred to as the \emph{diversified holding constraint}.  Later in Section~\ref{section: Empirical Studies}, we shall see that the diversified holding constraint can be somewhat replaced by imposing the ambiguity consideration in the return distribution.

\subsubsection{Survival Constraints.}

In practice, when considering investment leverage $L \geq 1$, a negative account value $V(t) < 0$ must be forbidden for all $t$ with probability one. This ensures the account remains \emph{survivable} and there is \emph{no bankruptcy}. The following lemma states sufficient conditions for ensuring a trade is survivable when the turnover cost is involved.


\begin{lemma}[Survivability Condition] \label{lemma: Weak Survivability Condition}
The probability  $
\mathbb{P} (V(t) \geq 0) = 1
$ for all~$t \geq 1$ 
if the following two conditions hold:
\begin{align} 
	\begin{cases}
		&	\sum_{i=1}^{n} K_i^{+}(t) |\min\{0, x_{i, \min}\}|- \sum_{i=1}^{n} K_i^{-}(t) \max\{0, x_{i, \max}\} \leq 1; \label{eq: survival constraint 1} \\
	&	| K(t) - K(t-1)|^\top C(t) \leq 1. 	  
	\end{cases}
\end{align}
\end{lemma}

\begin{proof} 
See~\ref{appendix: proofs in preliminaries}.
\end{proof}



\begin{definition}[Turnover Rate] \rm
For $t = 1, 2, \dots, T$, the \emph{turnover rate} of a portfolio at period~$t$ is defined as $TR(t) := |K(t) - K(t-1)|^\top \mathbf{1} = \sum_{i=1}^n |K_i(t) - K_i(t-1)|$. 
\end{definition}


\begin{remark}[Turnover Rate Constraint] \rm \label{remark: Turnover Rate Constraint}
Suppose the costs charged are the same for all assets, i.e., $c_i(t) = c \in [0,1]$ for all $i=1, \dots, n$. Then condition~\eqref{eq: survival constraint 1} implies that the turnover rate satisfies  $TR(t) < \frac{1}{c}$ for all $ t = 1, 2, \dots, T.$ 
Moreover, when considering the turnover cost rate $c$, we use $c_{\max}$ to denote the \emph{turnover cost limit}, i.e., we require that $|K(t) - K(t-1)|^\top C(t)\leq c_{\max} < 1$ for all $t$. This constraint restricts the weights to be adjusted at period~$t$. 
In the sequel, we record the totality of the trading constraints described above in the following definition.
\end{remark}

\begin{definition}[Totality of the Trading Constraints] \rm \label{definition: totality of the trading constraints}
Let $\mathcal{K}$ be the totality of the trading constraints, including short selling, leverage \eqref{eq: leverage constraint}, diversified holding \eqref{eq: diversification constraint}, and survival constraints~\eqref{eq: survival constraint 1} on portfolio weight $K$.
\end{definition}


\begin{remark}[Convexity and Compactness of $\mathcal{K}$] \rm
Note that short selling, leverage, diversified holding, and survival constraints are defined in $\mathbb{R}^n$ with linear inequalities. 
Each constraint forms a convex set.  Additionally, since each constraint set above is bounded and closed, the intersection of these sets, $\mathcal{K}$, is convex and compact.
\end{remark}

\subsection{Distributional Robust Optimal Portfolio}

For $t=1,2,\dots, T$, with $V(0) > 0$, let $U_t$ be a continuously differentiable and concave utility function. We consider the running objective
\begin{align}
J_p(t; K(t), K(t-1)) 
& := \mathbb{E}_p \left[ U_t \left( \frac{V(t)}{V(t-1)} \right)  \right] \notag \\
& = \mathbb{E}_p \left[ U_t \left(   (1 + K(t)^\top  X(t))(1 - |K(t) - K(t-1)|^\top C(t)) \right)  \right], \label{eq: running expected utility}
\end{align}
where $\mathbb{E}_p[\cdot]$ denotes the expectation operator with respect to the unknown probability distribution~$p \in S_m$. Here, $S_m$ is a probability simplex~set defined as
$
S_m := \left\{ p \in \mathbb{R}^m_+ : p^\top \textbf{1} = 1,\ p_j \geq 0,\ j=1,2,\dots ,m  \right\}
$
where $\textbf{1} \in \mathbb{R}^m$ is the one-vector and $\mathbb{R}^m_+ := \left\{x= [ x_1\ x_2 \ \cdots \ x_m ]^\top \in \mathbb{R}^m: x_i \geq 0, i=1,2,\dots,m \right\}$.
For notational simplicity, we may sometimes write $J_p(t)$ instead of $J_p(t;K(t),K(t-1))$. Some technical results related to the properties of the running objective are collected in~\ref{appendix: some technical results}.

Assume that the ambiguous return distribution set is of the convex polyhedral form 
$
\mathcal{P} := \{ p\in S_m: A_0p = d_0,\ A_1p \leq d_1\},
$
which is formed by finite linear inequalities and equalities where $A_0 \in \mathbb{R}^{m_0\times m},\ d_0 \in \mathbb{R}^{m_0},\ A_1 \in \mathbb{R}^{m_1 \times m}$ and $d_1 \in \mathbb{R}^{m_1}$.
Given $K(t-1) \in \mathcal{K}$, we seek to find the weight $K(t) \in \mathcal{K}$ that solves the distributional robust optimal portfolio problem 
\begin{align} \label{problem: distributional robust optimal portfolio problem }
	\max_{K(t) \in \mathcal{K}} \; \inf_{p\in \mathcal{P}} J_p(t;K(t),K(t-1))
\end{align}
for $t = 1, 2, \dots, T.$
The following result presents an equivalent optimization problem via duality theory.


\begin{theorem}[An Equivalent Distributional Robust Optimization Problem] \label{theorem: An Equivalent Distributional Robust Optimization Problem}
Let $t = 1,2,\dots,T$, given $K(t-1) \in \mathcal{K}$, the distributional robust optimal portfolio problem~\eqref{problem: distributional robust optimal portfolio problem } 
is equivalent to 
\begin{align} \label{problem: equivalent DRO}
	&\max_{K(t), \nu, \lambda} \; \min_{j} ( q(K(t)) + A_0^\top \nu + A_1^\top \lambda)_j -\nu^\top d_0 - \lambda^\top d_1 \\
	& \text{\rm s.t.} \; K(t) \in \mathcal{K}, \lambda \succeq 0 \notag
\end{align}
where 
$
q(K(t)) = [q(K(t))_1\; q(K(t))_2\; \cdots \; q(K(t))_m]^\top
$
with the $j$th component satisfying
\begin{align} \label{eq: q_K_t_j}
		q(K(t))_j = U_t \left( (1 + K(t)^\top x^j) ( 1 - |K(t) - K(t-1)|^\top C(t)) \right),\; j = 1, 2, \dots, m.
\end{align}
\end{theorem}

\begin{proof} 
See~\ref{appendix: proofs in preliminaries}.
\end{proof}

\begin{remark} \rm
It is worth noting that Theorem~\ref{theorem: An Equivalent Distributional Robust Optimization Problem} above generalizes the duality result in \cite[Theorem~2.1]{hsieh2024solving} to include the turnover cost. 
\end{remark}

\section{Extended Supporting Hyperplane Approximation}\label{section: Extended Supporting Hyperplane Approximation}

To facilitate computational efficiency in solving Problem~\eqref{problem: equivalent DRO} in practical large-scale portfolio optimization, this section significantly extends the supporting hyperplane approximation approach proposed in \cite{hsieh2024solving}.
While the original method addresses log-utility without transaction costs, our approach accounts for a general class of additively separable utilities and incorporates market friction arising from turnover transaction costs. This extension greatly broadens the applicability and practical relevance of the method.
Given the shorthand expressions $x := K(t)^\top X(t)$ and $c := |K(t) - K(t-1)|^\top C(t)$ and the bounds $x_{\min} \in (-1,0]$, $x_{\max}> 0$, $c_{\min}=0$, and $c_{\max} \in [0,1)$, define an auxiliary mapping $f_t: [x_{\min}, x_{\max}] \times [c_{\min}, c_{\max}] \to \mathbb{R}$ as 
\begin{align} \label{eq: log-function}
f_t(x, c) := U_t \left( ( 1 + x )( 1 - c ) \right).
\end{align}
for some $U_t$ that is a continuously differentiable, strictly monotonic, and concave utility function for~$t = 1, \dots, T$.

\begin{definition}[Additively Separable Utility] \rm \label{definition: additively separable utility}
For $t=1,2, \dots, T$, let $f_t(x, c) $ be defined as in~\eqref{eq: log-function}, where it is a continuously differentiable, strictly monotonic, and concave function. We say that~$f_t(x, c)$ is \emph{additively separable} in $x$ and~$c$ if there exist continuously differentiable functions $ \phi_{1,t}, \phi_{2,t}: \mathbb{R} \to \mathbb{R}$ and constants~$\alpha_t$ and~$\beta_t$ such that 
$\phi_{1,t}(x)$  is strictly concave and strictly increasing, and $\phi_{2,t}(c)$ is concave and strictly decreasing, and 
$$
f_t(x, c) = U_t((1+x)(1-c)) = \alpha_t \cdot \phi_{1,t}(x) + \beta_t \cdot  \phi_{2,t}(c)
$$
where $\alpha_t >0$ and $\beta_t >0$.
\end{definition}

\begin{example}[Illustration of Additively Separable Utilities] \rm \label{example: Illustration of Additively Separable Utilities}
The above definition is common. For example, if $ U_t $ takes a logarithmic form, $ U_t((1+x)(1-c)) = \gamma_t \log((1+x)(1-c)) $, where $ \beta_t $ is a time-dependent parameter, then we have $ \phi_{1,t}(x) = \log(1+x) $, $ \phi_{2,t}(c) = \log(1-c) $, and $\alpha_t = \beta_t = \gamma_t$. 
This gives us
\[
U_t((1+x)(1-c)) = \gamma_t \log((1+x)(1-c)) =  \alpha_t \cdot \phi_{1,t}(x) + \beta_t \cdot \phi_{2,t}(c).
\]
which assures the additive separability; see also Section~\ref{section: Successive Partition Points for Log-Additive Separable Utility} for further development with this log-additive separable utility and Section~\ref{section: Empirical Studies} for large-scale empirical studies.
As the second example, if $ U_t $ takes a power form, $ U_t((1+x)(1-c)) = \gamma_t [(1+x)^\delta + (1-c)^\delta ]$, where $ \gamma_t $ is a time-dependent parameter and $ \delta \in (0, 1)$. Then taking \( \phi_{1,t}(x) = (1+x)^\delta \), $ \phi_{2,t}(c) = (1-c)^\delta $, $ \alpha_t = \beta_t = \gamma_t $ yields
$
	 U_t((1+x)(1-c)) = \gamma_t [(1+x)^\delta + (1-c)^\delta ] =\alpha_t \cdot \phi_{1,t}(x) + \beta_t \cdot \phi_{2, t}(c).
$
Finally, if $ U_t $ takes the form of a Constant Relative Risk Aversion (CRRA) utility function, i.e.,
$
U_t((1+x)(1-c)) = \gamma_t \left( \frac{(1+x)^{1-\vartheta_t}}{1-\vartheta_t} + \frac{(1-c)^{1-\vartheta_t}}{1-\vartheta_t} \right),
$
where $ \gamma_t $ and $\vartheta_t$ are time-dependent parameters with $\vartheta_t > 1$. Then, by taking $ \phi_{1,t}(x) = \frac{(1+x)^{1-\vartheta_t}}{1-\vartheta_t}$, $\phi_{2,t}(c) = \frac{(1-c)^{1-\vartheta_t}}{1 - \vartheta_t}$, and $ \alpha_t = \beta_t =  \gamma_t $, it gives us
\[
U_t((1+x)(1-c)) =  \gamma_t \left( \frac{(1+x)^{1-\vartheta_t}}{1-\vartheta_t} + \frac{(1-c)^{1-\vartheta_t}}{1- \vartheta_t} \right) = \alpha_t \cdot \phi_{1,t}(x) + \beta_t \cdot \phi_{2, t}(c).
\]
All three utility functions satisfy the requirements of being continuously differentiable, strictly monotonic, and concave. \hfill \qedsymbol  
\end{example}

In the sequel, we shall assume that the utility function  $f_t(x,c)= U_t((1+x)(1-c))$  is additively separable in $ x $ and $ c $ for each time $ t $.
Having defined this, we are now ready to approximate the function $f_t$ by the hyperplanes derived from partition points of the intervals $[x_{\min}, x_{\max}]$ and $[0, 1)$. 
In particular, take the partitions $\{x_{l}\}_{l=0}^{M_x}$ and $ \{c_{r}\}_{r=0}^{M_c}$ for $l=0, \dots, M_x$ and $r = 0, \dots, M_c$~with 
\begin{align*}
x_{\min} = x_0 < x_1 < \dots < x_{M_x} = x_{\max} \ \text{ and } \
c_{\min} = c_0 < c_1  < \dots < c_{M_c} = c_{\max}.
\end{align*}  
Then, for $l = 0,1, \dots, M_x$ and $r = 0,1, \dots, M_c$, the hyperplanes are of the form
\begin{align}\label{eq: hyperplane expression}
h_{l, r}(x, c) := [a_l \; b_r] \begin{bmatrix}
	x \\ c
\end{bmatrix} + \gamma_{l,r}
\end{align}
To determine the coefficients $a_l, b_r,$ and $\gamma_{l ,r}$, the hyperplane must match the function $f(x, c)$ at each partition point $(x_l, c_r)$. That is, we require
$
h_{l, r}(x_l, c_r) = f_t(x_l, c_r),
$
which implies that
$
a_l x_l + b_r c_r + \gamma_{l, r} = \alpha_t \cdot \phi_{1, t}(x_l) + \beta_t \cdot \phi_{2, t}(c_r). 
$
Rearranging this expression yields the intercept coefficient $\gamma_{l, r}$.
Now, to match the slope of the function at each partition point, we calculate the partial derivatives of~$f_t$. 
With these derivatives, we can match the slopes of the hyperplane:
\begin{align*}
a_l  = \frac{\partial f_t}{\partial x} (x_l, c_r) = \alpha_t \cdot \phi_{1,t}'(x_l)  
\text{ and }
b_r = \frac{\partial f_t}{\partial c} (x_l, c_r) = \beta_t \cdot \phi_{2,t}'(c_r). 
\end{align*}
Using the values of $a_l$ and $b_r$, we solve for $\gamma_{l, r}$ using the function value at the partition point, which yields
$
\gamma_{l, r} = \alpha_t \cdot  \phi_{1,t}(x_l) + \beta_t \cdot \phi_{2,t}(c_r) - a_l x_l - b_r c_r.
$

For example, suppose that $f_t(x, c) = f(x, c) = \log(1+x) + \log (1-c)$. Then the associated hyperplane~$h_{l, r} $ is of the form of \eqref{eq: hyperplane expression} with coefficients
$a_l = \frac{1}{ 1 + x_l}, b_r = -\frac{1}{1 - c_r}$ and $\gamma_{l, r} = f(x_l, c_r) - a_l x_l-b_r c_r$.


\subsection{Robust Linear Program Formulation}

We now apply the idea of the supporting hyperplanes above to the distributional robust portfolio optimization problem~\eqref{problem: equivalent DRO} stated in Section~\ref{section: Preliminaries}.  Specifically, by taking $x := K(t)^\top x^j$ and $c := |K(t) - K(t-1)|^\top C(t)$, we set the associated minimum and maximum points for $x$ and $c$ as follows:
$x_{\min} = \min_j K(t)^\top x^j$, $x_{\max} = \max_j K(t)^\top x^j$,
$c_{\min} = 0$, and $c_{\max}$ is defined in Remark~\ref{remark: Turnover Rate Constraint}.
Then, for~$j=1,\dots,m$, we define 
$
q_j(K(t)) := U_t \left( (1 + K(t)^\top x^j ) ( 1 - |K(t) - K(t-1)|^\top C(t)) \right).
$
We can now approximate it via these hyperplanes~\eqref{eq: hyperplane expression} as follows:
\begin{align} 
q_j(K(t)) 
&\approx   \min_{l,r} \left\{ h_{l,r} \left( K(t)^\top x^j, \ |K(t)-K(t-1)|^\top C(t) \right) \right\} \label{eq: q_j approximation}\\
& = \min_{l,r} \left\{ a_l (K(t)^\top x^j)+ b_r (|K(t) - K(t-1)|^\top C(t)) + \gamma_{l,r} \right\} \nonumber
\end{align}
where $x^j := [x_1^j\; x_2^j \cdots \; x_n^j]^\top$ for $j = 1, 2, \dots, m$.
Additionally, the distributional robust portfolio optimization problem~\eqref{problem: distributional robust optimal portfolio problem } can be reformulated as the following robust linear program.

\begin{problem}[Approximate Robust Linear Program]\label{problem: ELG approxiamtion} \rm
For $t = 1, 2, \dots, T$, take $Z_j, W \in \mathbb{R}$, $K^+_i(t) > 0$ and $K_I^-(t) < 0$.  Then, the equivalent distributional robust optimal problem problem~\eqref{problem: equivalent DRO} can be approximated by the following robust linear~program: 
{\small	\begin{align*}
	& \max_{K(t), \nu, \lambda}\ W - \nu^\top d_0 - \lambda^\top d_1\\
	& {\rm s.t.} \;\;  \sum_{i=1}^n |K_i(t)| \leq L,\\
	& \qquad \sum_{i=1}^{n} K_i^{+}(t) |\min\{0, x_{i, \min}\}| - \sum_{i=1}^{n} K_i^{-}(t) \max\{0, x_{i, \max}\} \leq 1,\\
	&\qquad | K(t) - K(t-1))|^\top C(t) \leq 1,\\
	&\qquad \lambda \succeq 0,\\
	&\qquad Z_j \leq h_{l,r} \left( K(t)^\top x^j, \ |K(t)-K(t-1)|^\top C(t) \right), \; j = 1, \dots, m,\ l = 0,  \dots, M_x,\ r = 0, \dots, M_c,\\
	&\qquad W \leq Z_j + ( A_0^{\top} \nu + A_1^{\top} \lambda)_j.
\end{align*}
}	where
$
h_{l, r} \left( K(t)^\top x^j, \ |K(t) - K(t-1)|^\top C(t) \right) = a_l ( K(t)^\top x^j ) + b_r ( |K(t) - K(t-1)|^\top C(t)) + \gamma_{l, r}.
$
\end{problem}

\begin{remark} \rm
By taking zero turnover costs, i.e., $C(t) \equiv 0$, the robust linear program above reduces to the one considered in \cite{hsieh2024solving}.
\end{remark}

\subsection{Approximation Error Analysis} 
Below, we define the approximation error induced by the supporting hyperplane approach.  We shall then show that the total approximation error can be separated into the approximation errors for~$x$ and $c$, respectively.

\begin{definition}[Approximation Error Functions] \rm
For $t=1,2,\dots,T$, we denote $x := K(t)^\top x^j$ and $c := |K(t)-K(t-1)|^\top C(t)$.
Let $l = 0, 1, \dots, M_x$, and $r = 0, 1, \dots, M_c$. For $x \in [x_{\min}, x_{\max}]$ with $x_{\min} > -1$ and $c \in [c_{\min}, c_{\max}] \subseteq [0,1)$, we consider the mapping~$e: [x_{\min}, x_{\max}] \times [c_{\min}, c_{\max}] \to \mathbb{R}$ defined~by 
\begin{align} \label{eq: approximation error in total}
	e(x,c) := \left| f_t(x,c) - \min_{l,r} h_{l,r}(x, c) \right|,
\end{align}
which represents the approximation error between the hyperplanes and the objective function~$f_t$ defined in \eqref{eq: log-function}.
Moreover, we define $e_l: [x_{\min}, x_{\max}] \to \mathbb{R}$~by
\begin{align}\label{eq: approximation error along x-direction}
	e_l(x) := a_l (x - x_l) + \alpha_t \cdot \phi_{1,t}(x_l) - \alpha_t \cdot \phi_{1,t}(x),    
\end{align} 
and $e_r: [c_{\min}, c_{\max}] \to \mathbb{R}$ by
\begin{align}\label{eq: approximation error along c-direction}
	e_r(c) := b_r (c - c_r) + \beta_t \cdot \phi_{2,t}(c_r)  - \beta_t \cdot \phi_{2,t}(c).
\end{align}
where $a_l   = \alpha_t \cdot \phi_{1,t}'(x_l) $, $b_r = \beta_t \cdot \phi_{2,t}'(c_r)$, and 
$
	\gamma_{l, r}  = \alpha_t \cdot \phi_{1,t}(x_l) + \beta_t \cdot \phi_{2,t}(c_r)- a_l x_l - b_r c_r.
$	
\end{definition} 

Next, we examine the behavior of the approximation errors $e_l(x)$ and $e_r(c)$.

\begin{lemma}[Limiting Behavior and Monotonicity of Approximate Error] \label{lemma: Limit Behavior And Monotonicity of Approximate Error}
Let $\{x_l\}_{l=0}^{M_x}$ be a partition of $[x_{\min}, x_{\max}]$ for $l = 0,1,\dots,M_x$, and let $\{c_r\}_{r=0}^{M_c}$ be a partition of $[c_{\min}, c_{\max}]$ for $r = 0,1,\dots,M_c$. The following statements hold true.

\begin{itemize}
	\item[(i)] For $l \neq M_x$, the approximation error for $x$, $e_l(x)$, is strictly increasing in $(x_l, x_{\max}]$. For $l \neq 0$, the error $e_l(x)$ is strictly decreasing in $[x_{\min}, x_l)$. Additionally, $\lim_{x \to x_l} e_l(x) = 0$.
	
	\item[(ii)] For $l \neq M_c$, the approximation error for $c$, $e_r(c)$, is strictly increasing in $(c_r, c_{\max}]$. For $r \neq 0$, the error $e_r(c)$ is strictly decreasing in $[c_{\min}, c_r)$. Additionally, $\lim_{c \to c_r} e_r(c) = 0$.
\end{itemize}
\end{lemma}

\begin{proof} 
See Appendix~\ref{appendix: proofs in extended supporting hyperplane approximaiton}.
\end{proof}

With the aid of Lemma~\ref{lemma: Limit Behavior And Monotonicity of Approximate Error}, the following theorem indicates that the maximum approximation error $\sup_{x, c} e(x, c)$ induced by the hyperplane approximation approach is separable and can be represented as the sum of the approximation errors along the partitions for $x$ and $c$. 

\medskip
\begin{theorem}[Separable Maximum Approximation Error]\label{theorem: Separable Maximum Approximation Error}  
Let $h_{l,r}(x, c)$ be the hyperplanes defined in~\eqref{eq: hyperplane expression} that approximate $f_t(x, c) = U_t((1+x)(1-c))$. Then, the maximum approximation error is separable, i.e., 
\[
\sup_{x,c} e(x, c) = \sup_{x} \min_{l} e_l(x) + \sup_{c} \min_{r} e_r(c) 
\]
where $e_l(x)$ and $e_r(c)$ are defined in \eqref{eq: approximation error along x-direction} and \eqref{eq: approximation error along c-direction}, respectively, for $l = 0, 1, \dots, M_x$ and $r = 0, 1, \dots, M_c.$
\end{theorem}

\begin{proof} 
See~\ref{appendix: proofs in extended supporting hyperplane approximaiton}.
\end{proof}

\medskip
\begin{corollary} \label{corollary: Separable Maximum Approximation Error for Partitions}
For any partition $\{x_l\}_{l=0}^{M_x}$ and $\{c_r\}_{r=0}^{M_c}$, and for $x \in [x_p, x_{p+1}]$ and $c \in [c_q, c_{q+1}]$, where $p=0,1,\dots,M_x-1$ and $q=0,1,\dots,M_c-1$, the maximum approximation error over the subintervals is separable, i.e., 
\[
\sup_{ \substack{ x \in [x_p, x_{p+1}]\\ c \in [c_q, c_{q+1}]} } e(x, c) = \sup_{x \in [x_p, x_{p+1}]} \min_{l} e_l(x) + \sup_{c \in [c_q, c_{q+1}]} \min_{r} e_r(c).
\]
\end{corollary}

\begin{proof} 
See~\ref{appendix: proofs in extended supporting hyperplane approximaiton}.
\end{proof}


We now determine the hyperplanes of $e_l(x)$ and $e_r(c)$ from given $(x_{\min}, c_{\min})$ and compute the point that yields the corresponding maximum approximation~error.

\begin{lemma}[Characterization of Maximum Approximation Errors]  \label{lemma: Maximum Approximation Errors}
Fix partitions $\{ x_l \}_{l=0}^{M_x}$ and $\{ c_r \}_{r=0}^{M_c}$ such that $x_0 = x_{\min}, x_{M_x} = x_{\max}$, $c_0 = c_{\min}$ and $c_{M_r} = c_{\max}$. Given $x_p$ and $c_q$, where $p = 0, 1, \dots, M_x-1$ and $q = 0, 1, \dots, M_c-1$,  for $x \in [x_p, x_{p+1}]$ and $c \in [c_q, c_{q+1}]$, there exists a pair of functions $(x^{\prime}(x_{p+1}), c^{\prime}(c_{q+1}))$ such that 
\[    
	\sup_{ \substack{x \in [x_p, x_{p+1}] \\ c \in [c_q, c_{q+1}]} }  e(x, c)  
	= e_p( x^{\prime} ) + e_q( c^{\prime} ),
\]
where $	  x^{\prime}  = x^{\prime }(x_{p+1})   = \frac{	 \phi_{1,t}'(x_p)   x_p - \phi_{1,t}'(x_{p+1}) x_{p+1}  +   \phi_{1,t} (x_{p+1})  -  \phi_{1,t} (x_p)  }{\phi_{1,t}'(x_p) -  \phi_{1,t}'(x_{p+1}) }
$
and 
$
c^{\prime}  = c^{\prime}(c_{q+1})
= 	 \frac{	\phi_{2, t}'(c_q) c_q- 	\phi_{2, t}'(c_{q+1}) c_{q+1}  + \phi_{2, t}(c_{q+1}) - \phi_{2, t}(c_q)  }{ \phi_{2, t}'(c_q)  - \phi_{2, t}'(c_{q+1}) }.
$
\end{lemma}

\begin{proof} 
See~\ref{appendix: proofs in extended supporting hyperplane approximaiton}
\end{proof}

\subsection{Successive Partition Points} \label{section: Successive Partition Points for Log-Additive Separable Utility}
Given the pair $(x_p, c_q)$, the results in previous subsections are useful for determining the successive partition points, which leads to optimal number of hyperplanes.

\begin{theorem}[Successive Partition Points] \label{theorem: Successive Partition Points for general Additively Separable Utility}
	Fix partitions $\{ x_l \}_{l=0}^{M_x}$ and $\{ c_r \}_{r=0}^{M_c}$ such that $x_0 = x_{\min},  x_{M_x} = x_{\max}$, $c_0 = c_{\min}$ and $c_{M_r} = c_{\max}$. Given $x_p$ and $c_q$, where $p = 0, 1, \dots, M_x-1$ and $q = 0, 1, \dots, M_c-1$, let the error tolerance be $\varepsilon = \varepsilon_x+\varepsilon_c > 0$ for some $\varepsilon_x > 0$ and $\varepsilon_c > 0.$  Then the successive partition points $x_{p+1}$ and $c_{q+1}$ satisfy:
	$$
	x_{p+1} =  x_p + \mathcal{A}^* + \mathcal{B}^* 
	\text{ and }
c_{q+1} = c_q + \mathcal{D}^* + \mathcal{E}^*,
	$$
	where $ \mathcal{A}^*$ solves $ 	\frac{\varepsilon_x }{ \alpha_t }   = \phi_{1, t}'(x_p) \mathcal{A} - \phi_{1, t}(  \mathcal{A} + x_p ) + \phi_{1, t}(x_p) $ 
	and $\mathcal{B}^*$ solves $ 	\mathcal{B} \cdot (\phi_{1,t}'(x_p + \mathcal{A}^* + \mathcal{B}) ) =     \phi_{1,t} (x_p + \mathcal{A}^* + \mathcal{B})  -  \phi_{1,t} (x_p)  -	\mathcal{A}^* \phi_{1,t}'(x_p)  $,
	and
 $ \mathcal{D}^*$ solves $  	\frac{\varepsilon_c}{\beta_t}  = \phi_{2,t}'(c_q) \mathcal{D} +  \phi_{2,t} (c_q) -  \phi_{2, t} ( \mathcal{D} + c_q )$ 
and~$\mathcal{E}^*$ solves $ \mathcal{E} \cdot (\phi_{2, t}'(c_q + \mathcal{D}^* + \mathcal{E}) )=    \phi_{2, t} (c_q + \mathcal{D}^* + \mathcal{E})  -  \phi_{2, t} (c_q)  -	\mathcal{D}^* \phi_{2, t}'(c_q)  .  $ 
\end{theorem}

\begin{proof} 
	See~\ref{appendix: proofs in extended supporting hyperplane approximaiton}.
\end{proof}

An interesting special case arises if we consider the log-additive separable utility, i.e., with $\alpha_t = \beta_t = 1$ and $\phi_{1,t} = \log(1+x)$ and $\phi_{2, t}(c) = \log (1-c)$. 
The utility, , as seen in Example~\ref{example: Illustration of Additively Separable Utilities}, is given by  $U_t((1+x)(1-c)) := \log(1+x)  + \log(1-c)$. Given the partition pair $(x^{\prime},c^{\prime})$, the following result indicates that the successive partition points can be obtained recursively in a more compact format. 

\begin{theorem}[Successive Partition Points for Log-Additive Separable Utility]\label{theorem: Successive Partition}
Fix partitions $\{ x_l \}_{l=0}^{M_x}$ and $\{ c_r \}_{r=0}^{M_c}$ such that $x_0 = x_{\min},  x_{M_x} = x_{\max}$, $c_0 = c_{\min}$ and $c_{M_r} = c_{\max}$. Given $x_p$ and $c_q$, where $p = 0, 1, \dots, M_x-1$ and $q = 0, 1, \dots, M_c-1$, let the error tolerance be $\varepsilon = \varepsilon_x+\varepsilon_c > 0$ for some $\varepsilon_x > 0$ and $\varepsilon_c > 0.$  Then the successive partition points $x_{p+1}$ and $c_{q+1}$ satisfy:
$$
x_{p+1} = (1 + \mathsf{a}_x) x_p +  \mathsf{a}_x \text{ and }
c_{q+1} = (1 -  \mathsf{d}_c) c_q +  \mathsf{d}_c
$$
where $ \mathsf{a}_x$ solves $\frac{1 + \mathsf{a}}{\mathsf{a}}\log(1 + \mathsf{a}) = \mathsf{b}_x$ with $\mathsf{b}_x$ as the solution of $\mathsf{b} - \log\mathsf{b} - 1 = \varepsilon_x$ and $\mathsf{d}_c$ solves the nonlinear equation $\frac{1 - \mathsf{d}}{\mathsf{d}}\log\left( \frac{1}{1 - \mathsf{d}} \right) = \theta_c$ with $\theta_c$ as the solution of $\theta - \log\theta-1=\varepsilon_c$.
\end{theorem}

\begin{proof} 
	See~\ref{appendix: proofs in extended supporting hyperplane approximaiton}.
\end{proof}

By the above process, we separate the maximum hyperplane approximation error into two error functions $e_l(x),\ l = 0,1,\dots, M_x$ and $e_r(c),\ r = 0,1,\dots, M_c$, and recursively construct hyperplanes~$h_{l,r}$. Then we obtain the following result:

\begin{lemma}[Optimal Number of Hyperplanes] \label{lemma: Optimal Number Of Hyperplanes}  
Given the maximum error tolerance constant~$\varepsilon > 0$, the corresponding optimal number of hyperplanes required is given by
$
M := M_x + M_c,
$
where $M_x$ and $M_c$ are the optimal numbers of hyperplanes for~\eqref{eq: approximation error along x-direction} and~\eqref{eq: approximation error along c-direction}, respectively.
\end{lemma}

\begin{proof} 
See~\ref{appendix: proofs in extended supporting hyperplane approximaiton}.
\end{proof}

\section{Empirical Studies: Large-Scale Robust Portfolio Management} \label{section: Empirical Studies}
This section provides an extensive empirical study using large-scale historical price data to substantiate our theory. Throughout this section, we adopt the log-additive separable utility, which aligns with the standard ELG theory. For further reference, see \cite{MacLean_Thorp_Ziemba_2011book, Cover_Thomas_2012, rujeerapaiboon2016robust}, and \cite{hsieh2023asymptotic, hsieh2024solving}.

\emph{Data.}
We use daily adjusted closing price data for~S\&P 500 constituent stocks from~\cite{YahooFinance}, covering three years from January 1, 2021, to December 31, 2023.  The constituent stocks of the~S\&P~500 index may change over time, leading to an incomplete dataset. 
To address issues with missing values, we consider only stocks that were not replaced or added to the index during this period. 
Any remaining missing values are completed using linear interpolation, e.g., see~\cite{Newbury1981}. 
As a result, our dataset includes $477$ individual stocks and one additional risk-free asset, covering a total of $753$ trading days. 
The risk-free asset has an annualized interest rate of $r_f:= 0.02$, which serves as a reasonable approximation given the variation in U.S. Treasury yields during the~period.

\emph{The Benchmarks.}
To compare with the classical ELG portfolio, we use the log-additively separable utility with hyperplane approximation approach described in Section~\ref{section: Extended Supporting Hyperplane Approximation}. This ensures consistency in the comparison, as the optimal solution obtained from the hyperplane approach can be shown to be arbitrarily close to the optimal solution obtained via ELG when there is no ambiguity and no cost; see \cite{hsieh2024solving}.
Additionally, we consider two standard performance benchmarks: the buy-and-hold strategy on SPY, an ETF that tracks the S\&P 500 index, and the equal-weight buy-and-hold portfolio of S\&P 500 index constituents. For convenience, we shall use the shorthand \texttt{ELG} for the expected log-growth portfolio,  \texttt{HYP} for the hyperplane approximation approach with log-additively separable utility, \texttt{SPY} for the SPY ETF, and \texttt{EW} for the equal-weight buy-and-hold~portfolio. 

\emph{Simulation Details and Parameter Settings.} 
Through the following experiments, we solve Problem~\ref{problem: ELG approxiamtion} using the sliding window method, e.g., see \cite{wang2022data}. Specifically, we rebalance the portfolio quarterly, and in each rebalance, we optimize the portfolio with six-month training data that ends one day before the rebalance. 
With error tolerance constants $\varepsilon_x = 0.001$ and~$\varepsilon_c = 1 \times 10^{-5}$,
the corresponding number of supporting hyperplanes are $M_x \in \{6, 7, \dots, 11\}$ and~$M_c \in \{1, 2\}$. The number of hyperplanes varies since the data used in each rebalancing has different price volatilities.

For the backtests, we use a leverage constant of $L = 1.5$. 
Moreover, we set different cost rates to observe the effects of turnover costs on portfolio performance. Specifically, we consider $c \in \{0, 0.001, 0.005, 0.01\}$.
For nonzero cost rates, we set two constraints on the turnover cost limit: $c_{\max} \in \{c \cdot 2L, c \cdot \frac{L}{2}\} $ with $c \neq 0$.  
The choice of $2L$ corresponds to the scenario where an asset is invested with leverage, sold upon rebalancing, and a completely new asset is purchased with leverage, effectively contributing $2L$.
On the other hand, the choice of $\frac{L}{2}$ represents a more conservative scenario where only a partial turnover occurs, allowing for a more gradual portfolio adjustment.
These two choices allow us to analyze the impact of different levels of turnover constraints on the portfolio.
It is readily verified that the constraint $|K(t) - K(t-1)|^{\top} C < c_{\max} < 1$, as described in Section~\ref{subsection: Trading Constraints}, is satisfied. Specifically, we consider two turnover rate constraints: $TR(t) < 2L$ and $TR(t) < \frac{L}{2}$.
Lastly, since the turnover costs are satisfied for any $c_{\max}$, we set $c_{\max} = 0.01$ for the case of zero transaction cost rate.

\subsection{Performance Analysis}
Table~\ref{table: Performance Metrics of SPY and Equal-weight} summarizes the benchmark trading performances of the buy-and-hold strategy on both \texttt{SPY} and \texttt{EW}, with three-year returns and annual Sharpe ratio. Since transaction costs are only charged for the initial investment, the maximum drawdown for the four transaction cost rates is similar. 
From the table, it is clear that the cumulative return and the Sharpe ratio decline as the transaction cost rate increases. For a zero cost rate, the cumulative return of \texttt{SPY} is about $0.153$, slightly higher than that of \texttt{EW}, which is~$0.131$.
In the sequel, we shall compare the trading performance above with our proposed approach \texttt{HYP} and the classical \texttt{ELG} portfolio.

\begin{table}[h!]
\centering
\footnotesize
\caption{Benchmark Performance: \texttt{SPY} and \texttt{EW} }

\begin{tabular}{l c c c c   c c c c c}
	\toprule
	& \texttt{SPY} &  &  &  &  \texttt{EW} &  &  &  \\
	\midrule
	Cost Rate $c$ & 0.0 & 0.001 & 0.005 & 0.01 & 0.0 & 0.001 & 0.005 & 0.01 \\
	Cumulative Return & 0.153 & 0.152 & 0.148 & 0.142 & 0.131 &  0.130 & 0.125 & 0.120\\
	Max Drawdown & 0.245 & 0.245 & 0.245 & 0.245 & 0.209 & 0.209 & 0.209 & 0.209\\
	Annualized Sharpe Ratio& 0.295 & 0.292 & 0.284 & 0.273 & 0.252 & 0.249 & 0.241 & 0.230\\
	\bottomrule
\end{tabular}

\label{table: Performance Metrics of SPY and Equal-weight}
\end{table}

\subsubsection{Performance Analysis Without Ambiguity ($\gamma = 0$).}
We first consider the nominal case where $\gamma = 0$, i.e., the return has no ambiguity. Then, we optimize portfolios of $n = 478$ assets, comprising $477$ constituent stocks of the S\&P 500 and one risk-free asset, by solving Problem~\ref{problem: ELG approxiamtion}.
Table~\ref{table: Performance Metrics with diversification} summarizes the out-of-sample performance of portfolios when $\gamma = 0$. Since we are recomputing weights, turnover rate, and optimal value every three months, the table reports the \emph{average} invested weight in risky assets, \emph{average} optimal value for the optimal values, and \emph{average} turnover~rate.

\begin{table}[h!]
\footnotesize
\caption{Trading Performance with Diversified Holding Constraint and $\gamma = 0$}
\centering
\begin{tabular}{l c c c c c}
	\toprule[1.2pt]
	Case of  \texttt{ELG} portfolio & 2.1 & 2.2 & 2.3 & 2.4 & 2.5 \\
	\midrule
	Cost Rate $c$ & 0 & 0.001 & 0.001 & 0.005 & 0.005 \\
	Turnover Cost Limit $c_{\max}$ & 0.01 & 0.003 & 0.00075 & 0.00755 & 0.00375 \\
	Cumulative Return & 0.114 &  0.062 & 0.0599 & 0.089 &  0.089 \\
	Max Drawdown & 0.167 & 0.202 & 0.194 & 0.007 & 0.007 \\
	Annualized Sharpe Ratio & 0.230 & 0.104 & 0.096 & 1.447 & 1.441 \\
	Average Turnover Rate & 0.59 & 0.36 & 0.34 & $7.13\mathrm{e}{-3}$ & $7.34\mathrm{e}{-3}$ \\
	Average Invested Weight & 0.88 & 0.86 & 0.81 & 0.03 & 0.03 \\
	Average Optimal Value & $1.22\mathrm{e}{-3}$  & $6.92\mathrm{e}{-4}$ & $6.76\mathrm{e}{-4}$ & $3.86\mathrm{e}{-5}$ & $3.83\mathrm{e}{-5}$ \\
	Average Running Time (sec) & 6,031.13 & 8,697.77 & 6,283.13 & 4,878.52 & 5,680.08 \\
	\midrule[1.0pt]
	Case of \texttt{HYP} portfolio & 2.6 & 2.7 & 2.8 & 2.9 & 2.10 \\
	\midrule
	Cost Rate $c$ & 0 & 0.001 & 0.001 & 0.005 & 0.005 \\
	Turnover Cost Limit $c_{\max}$ & 0.01 & 0.0015 & 0.00075 & 0.0075 & 0.00375 \\
	Cumulative Return & 0.103 & 0.058 & 0.061 & 0.087 & 0.090 \\
	Max Drawdown & 0.176 & 0.206 & 0.199 & 0.008 & 0.008 \\
	Annualized Sharpe Ratio & 0.202 & 0.094 & 0.0992 & 1.2582 & 1.3621 \\
	Average Turnover Rate & 0.60 & 0.35 & 0.33 & 0.01 & 0.01 \\
	Average Invested Weight & 0.89 & 0.86 & 0.82 & 0.04 & 0.04 \\
	Average Optimal Value & $1.57\mathrm{e}{-3}$ & $1.05\mathrm{e}{-3}$ & $1.03\mathrm{e}{-3}$ & $3.8\mathrm{e}{-4}$ & $3.61\mathrm{e}{-4}$ \\
	Average Running Time (sec) & 8.11 & 4.64 & 4.36 & 7.59 & 4.16  \\
	\bottomrule
\end{tabular}

\label{table: Performance Metrics with diversification}
\end{table}


In the table, compared to the running time of \texttt{ELG} ranging from 5,000 to 9,000 seconds, the longest running time for our \texttt{HYP} approach is less than 10 seconds, significantly improving computational efficiency with the supporting hyperplane approximation. 
For example, see Case~2.2 in Table~\ref{table: Performance Metrics with diversification}, \texttt{ELG} optimization takes the longest running time in the table, which is $8,697.77$ seconds, while the running time of \texttt{HYP} is only $4.64$ seconds. Moreover, the performance of \texttt{HYP} and \texttt{ELG} is similar, as seen in the row of cumulative return and maximum drawdown. 


Additionally, Table~\ref{table: Performance Metrics with diversification} clearly shows that the trading performance is affected by transaction cost. For example, from Cases 2.1 to 2.2 in \texttt{ELG} portfolio, when the transaction cost rate rises from zero to $0.001$, the cumulative return decreases from $0.114$ to $0.062$. Similarly, in Cases~2.6 and 2.7, the cumulative return of the corresponding \texttt{HYP} portfolio also decreases from $0.103$ to $0.058$. As a result, the Sharpe ratio also decreases in these cases. 
However, if we increase the transaction cost further to~$c=0.005$, the cumulative returns for both \texttt{ELG} and \texttt{HYP} increase, as seen in Cases 2.4 and 2.9 of Table~\ref{table: Performance Metrics with diversification}. These increases in cumulative return result from a portfolio concentration on the risk-free asset, where the average weight invested in risky assets declines to~$0.03$ and $0.04$, respectively. 

As a further illustration,
Figure~\ref{fig: D_Gamma0_LTurnover_elgVShyp} depicts the account value of the four portfolios (\texttt{SPY, EW, ELG, HYP}) where the turnover cost limit is $c_{\max} := 2L$. The gray dashed vertical lines denote the day we rebalance \texttt{ELG} and \texttt{HYP} portfolios. When the transaction cost rate is relatively large, e.g.,~$c = 0.005$, we see that the invested weight concentrates on the risk-free asset, leading to almost linear account~growth.

\begin{figure}[h!]
\centering
\includegraphics[width=.7\linewidth]{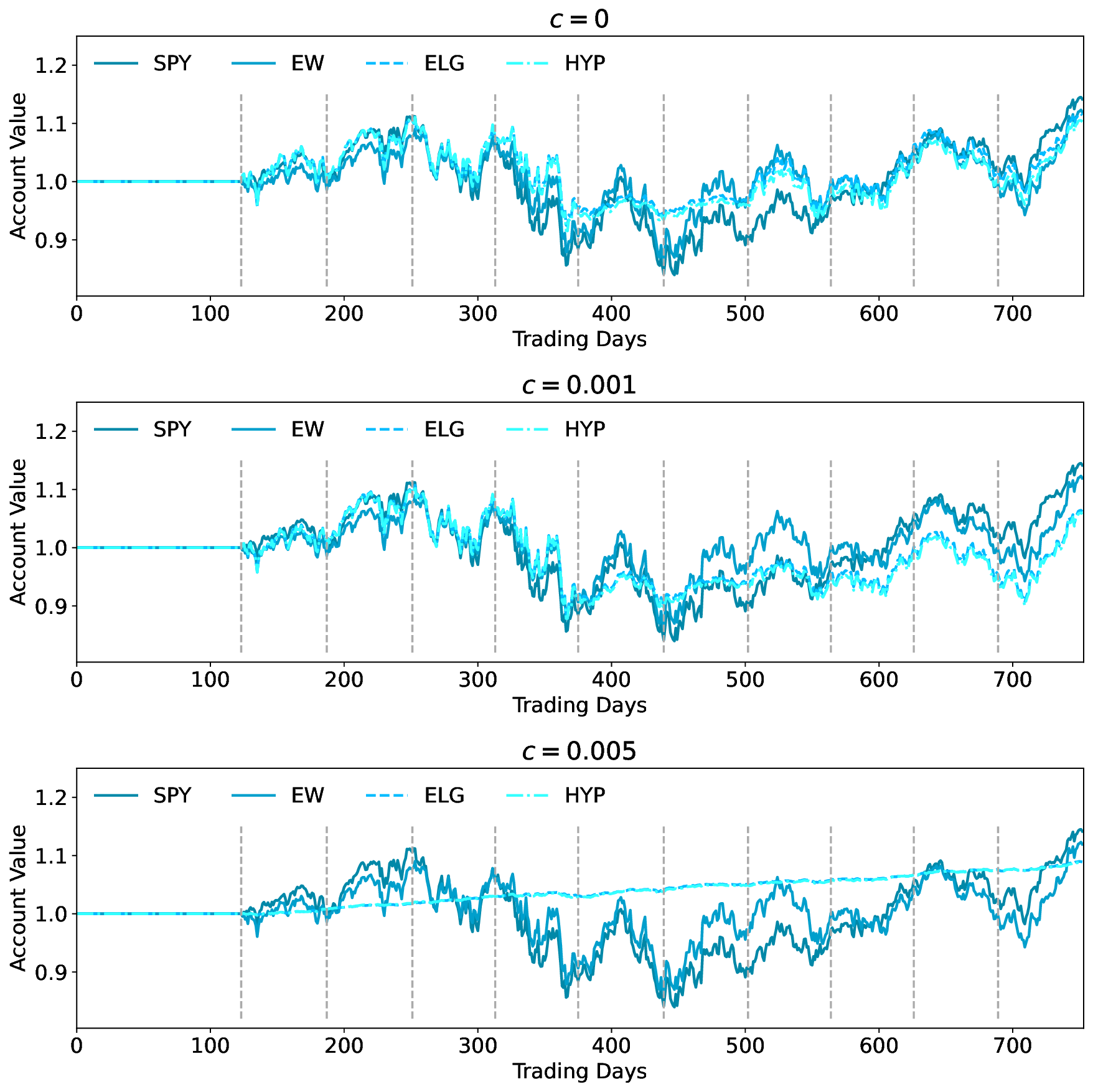}  
\caption{Trading Performance of Four Portfolios: \texttt{SPY, EW, ELG} and \texttt{HYP} with Ambiguity Constant $\gamma = 0$, Turnover Cost Limit $c_{\max} = 1.5$, and Various Cost Rates $c \in \{0, 0.001, 0.005\}$. }  
\label{fig: D_Gamma0_LTurnover_elgVShyp}
\end{figure}

\subsubsection{Performance Analysis with Ambiguity Considerations  ($\gamma > 0$).}

To examine the effect of the ambiguity constant $\gamma \in (0,1)$ on the robustness of the \texttt{HYP} portfolio, we solve Problem~\ref{problem: ELG approxiamtion} using the same data and rebalancing frequency as previously, without the diversified holding constraint. Table~\ref{table: Performance Metrics with different gamma} reports the trading performance of portfolios under different $\gamma$, with $c = 0.001$ and $c_{\max} = 0.003$. We observe that portfolios have higher returns when $\gamma$ is lower, while volatility is higher. 
For example, in Cases 3.1 and 3.6 of the table, the returns are approximately $1.939$ and~$0.175$, respectively. However, the Sharpe ratios are about $0.604$ and $0.714$. More interestingly, in a bear market regime such as the year 2022, which corresponds to the fifth to eighth gray lines in Figure~\ref{fig: DifferentGamma2024}, \texttt{HYP} portfolios optimized with higher $\gamma$ tend to allocate more weight to the risk-free asset, leading to more stable account values.

\begin{table}[h!]
\tiny
\caption{Trading Performance of \texttt{HYP} Portfolio with Different $\gamma$}
\centering
\begin{tabular}{l c c c c c c c}
\toprule[1.2pt]
Turnover Cost Limit $c_{\max} = 0.003$ & 3.1 & 3.2 & 3.3 & 3.4 & 3.5 & 3.6\\
\midrule
Ambiguity Constant $\gamma$ & 0.0 & 0.1 & 0.2 & 0.3 & 0.4 & 0.5\\
Cumulative Return & 1.939 & 0.558 & 0.551 & 0.332 & 0.275 & 0.175 \\
Max Drawdown & 0.327 & 0.460 & 0.356 & 0.303 & 0.239 & 0.092 \\
Annualized Sharpe Ratio & 0.604 & 0.586 & 0.6584 & 0.505 & 0.491 & 0.714 \\
Average Turnover Rate 
& 1.31 & 1.35 & 1.24 & 1.06 & 0.85 & 0.30\\
Average Invested Weight & 1.50 & 1.40 & 1.27 & 1.08 & 0.85 & 0.22\\
Average Optimal Value & $5.19\mathrm{e}{-3}$ & $2.74\mathrm{e}{-3}$ & $1.42\mathrm{e}{-3}$ & $5.84\mathrm{e}{-4}$ & $6.11\mathrm{e}{-5}$ & $1.26\mathrm{e}{-4}$\\
Average Running Time (sec) & 30.91 & 27.86 & 19.28 & 23.64 & 18.94 & 12.66\\

\bottomrule
\end{tabular}

\label{table: Performance Metrics with different gamma}
\end{table}

\begin{figure}[h!]
\centering
\includegraphics[width=.7\linewidth]{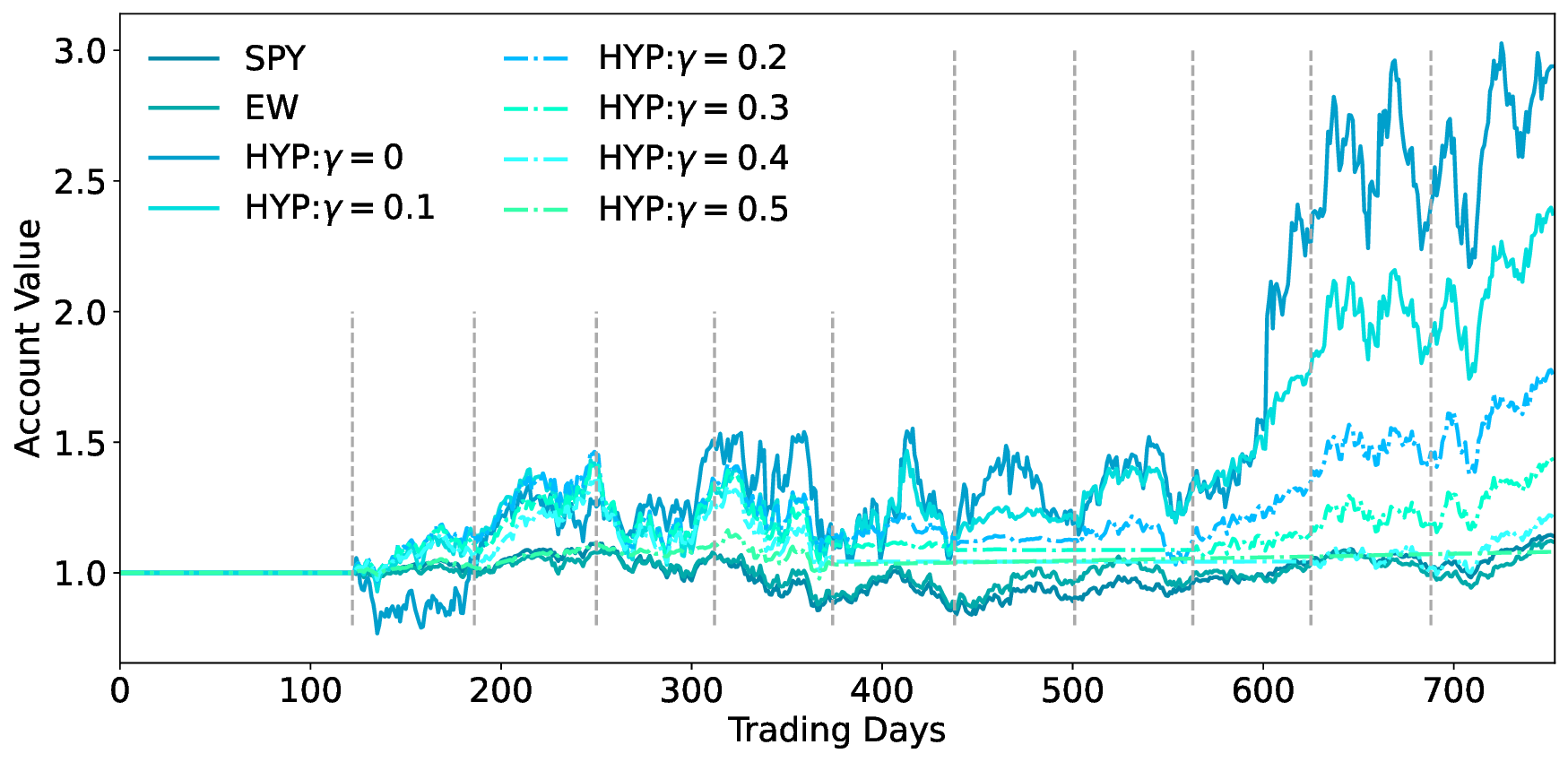}  
\caption{Account Values of \texttt{HYP} under Different $\gamma$, without Diversified Holding Constraint.}  
\label{fig: DifferentGamma2024}
\end{figure} 

Additionally, Figure~\ref{fig: ERoverC} illustrates the relationship between transaction cost rate and the expected return under different ambiguity constant $\gamma \in \{0, 0.1,\dots, 0.7\}$. Without the diversified holding constraint, the results shown in the figure are consistent with Table~\ref{table: Performance Metrics with diversification}, where returns decrease as the transaction cost rate rises. Moreover, the expected return at zero transaction cost decreases as the ambiguity constant $\gamma$ increases.

\begin{figure}[h!]
\centering
\includegraphics[width=.7\linewidth]{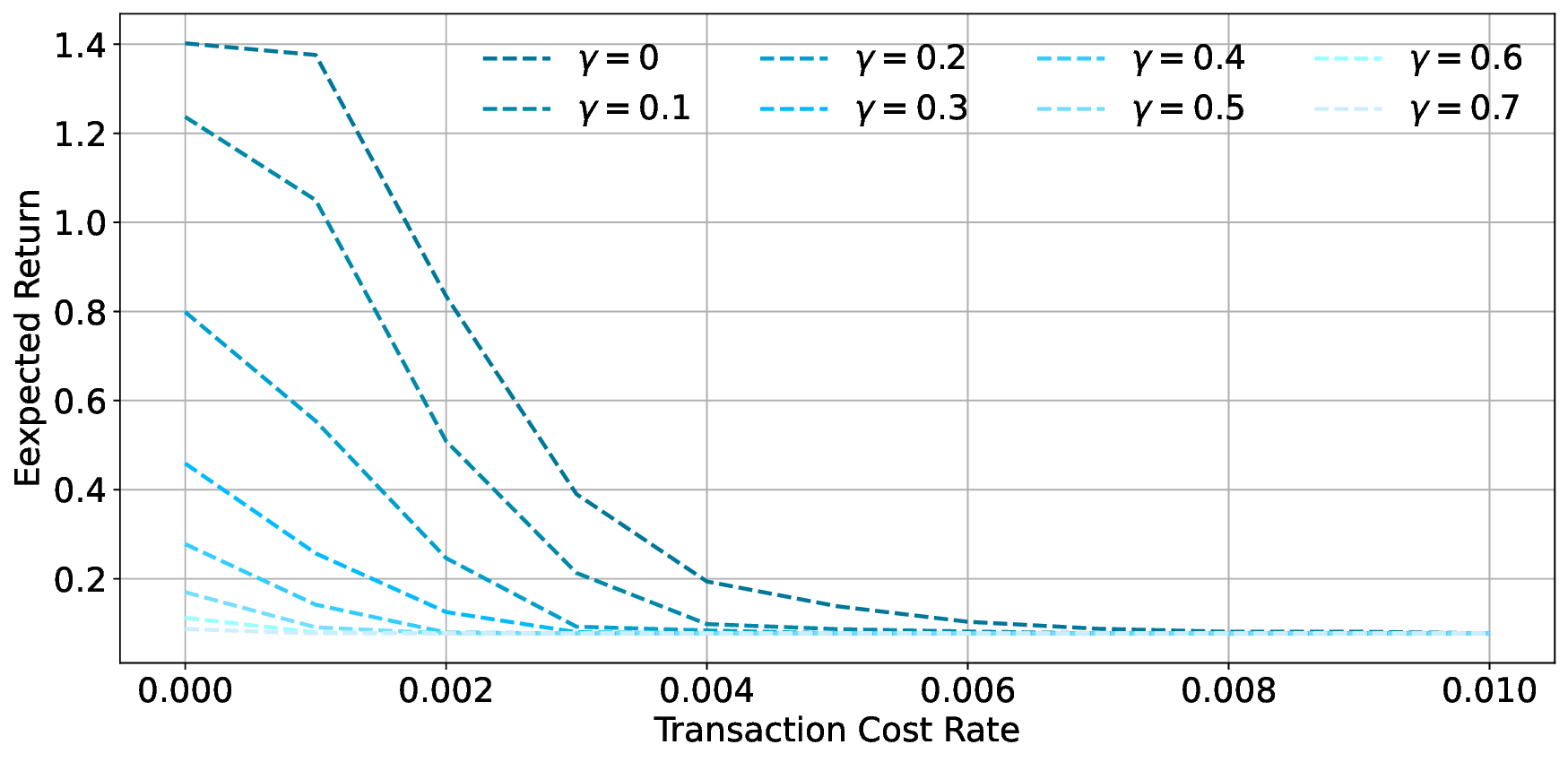}  
\caption{Expected Return of \texttt{HYP} Portfolio with Various $\gamma$.}  
\label{fig: ERoverC}
\end{figure}

\subsection{Diversification Effects via Ambiguity Constant}
As seen in the previous section, increasing the ambiguity constant may suggest a tendency towards diversification. This section further studies the relationship between diversification effects and ambiguity constants $\gamma$.
Specifically, we consider the \texttt{HYP} portfolio and impose the constraint that the sum of weights invested in risky assets only and the leverage equals $L = 1$, i.e., a cash-financed case, to observe whether portfolios tend to follow a diversified equal weight $\frac{1}{n}$ strategy. Additionally, we set the transaction cost rate $c = 0$ and rebalance the portfolio yearly with in-sample data from a previous year. 


Figure~\ref{fig: Diversity@Gamma} shows the maximum weight $\max_i K_i(t)$ in each rebalance for various $\gamma = 0, 0.1, \dots, 1$. We see that portfolios become more diversified as the ambiguity constant~$\gamma$ increases. However, an interesting finding is that the maximum weight does not converge to $\frac{1}{n}$ strategy as the ambiguity constant increases,  ending at $\max_i K_i(t) \approx 0.11$.  We envision that this is due to the polyhedral structure of the ambiguity set. However, it is beyond the scope of this paper, and we shall leave further study of this for future work.

\begin{figure}[h!]
\centering
\includegraphics[width=.7\linewidth]{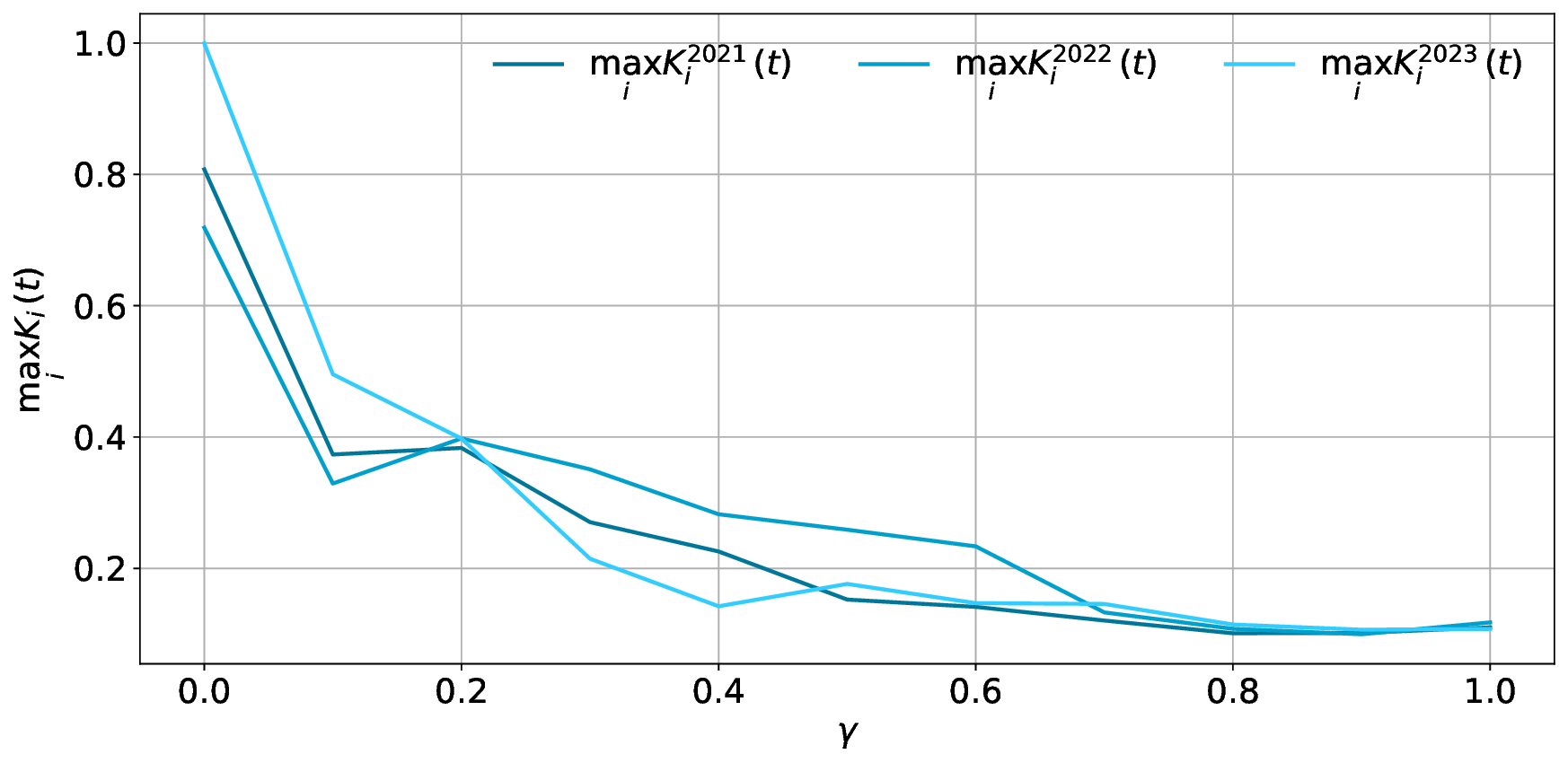}  
\caption{Diversification Effects via Ambiguity Constant in \texttt{HYP}: The Maximum Portfolio Weight $\max_i K_i^{Y}(t)$ Versus $\gamma$ where $\max_i K_i^{Y}(t)$ is the Maximum Portfolio Weight of the $Y \in \{2021, 2022, 2023\}$ Year.}  
\label{fig: Diversity@Gamma}
\end{figure}

\subsection{Optimal Number of Hyperplanes}
Here we show that the the \texttt{HYP} portfolio can be constructed using an efficient optimal number of hyperplanes for a large-scale case. 
Specifically, we calculate the estimation errors $e_l(x)$ and $e_r(c)$ in the S\&P 500 case. The corresponding parameters are set as follows: $\varepsilon := 2$, $c := 0.01$, $c_{\min} := 0$ and $c_{\max} := 0.02$. The values for  $x_{\min}$ and $x_{\max}$ are derived from historical returns used in the previous section. 
Table~\ref{table: Approximation Error for different parameters} summarizes the approximation errors for different  choices of $\varepsilon_x$ and~$\varepsilon_c$.

As seen in the table, the approximation errors~$e_l(x)$ and $e_r(c)$ are precisely controlled under $\varepsilon_x$ and $\varepsilon_c$ using the required number of hyperplanes $M_x$ and $M_c$. However, for example, if we remove any hyperplane that is neither the first one nor the last, the maximum approximation error for~$e_l(x)$ exceeds~$\varepsilon_x$.  Consequently, the total approximation error $e(x, c)$ exceeds the tolerance constant~$\varepsilon$.   A similar observation is made for $\varepsilon_c$.

\begin{table}[h!]
\centering
\footnotesize
\caption{Optimal Number of Hyperplanes: A Large-Scale Revisit}

\begin{tabular}{c c c c c c}
	\toprule
	$\varepsilon_x$ & $\varepsilon_c$ & \# of hyperplanes & $\displaystyle{\sup_x \min_l} \ e_l(x)$ & $\displaystyle{\sup_c \min_r} \ e_r(c)$ &  Error Violations\\
	\midrule
	$1\mathrm{e}{-5}$ & $1\mathrm{e}{-5}$ & $(M_x,\ M_c)$ & $1\mathrm{e}{-5}$ & $1\mathrm{e}{-5}$ & No \\
	$1\mathrm{e}{-5}$ & $1\mathrm{e}{-5}$ & $(M_x-1,\ M_c)$ & $4\mathrm{e}{-5}$ & $1\mathrm{e}{-5}$ & Yes \\
	$1\mathrm{e}{-5}$ & $1\mathrm{e}{-5}$ & $(M_x,\ M_c-1)$ & $1\mathrm{e}{-5}$ & $4\mathrm{e}{-5}$ & Yes \\
	$1\mathrm{e}{-5}$ & $1\mathrm{e}{-5}$ & $(M_x-1,\ M_c-1)$ & $4 \mathrm{e}{-5}$ & $4\mathrm{e}{-5}$ & Yes \\ 
	\midrule[1.0pt]
	$1.5\mathrm{e}{-5}$ & $5\mathrm{e}{-6}$ & $(M_x,\ M_c)$ & $1.5\mathrm{e}{-5}$ & $5\mathrm{e}{-6}$ & No \\
	$1.5\mathrm{e}{-5}$ & $5\mathrm{e}{-6}$ & $(M_x-1,\ M_c)$ & $6\mathrm{e}{-5}$ & $5\mathrm{e}{-6}$ & Yes \\
	$1.5\mathrm{e}{-5}$ & $5\mathrm{e}{-6}$ & $(M_x,\ M_c-1)$ & $1.5\mathrm{e}{-5}$ & $2\mathrm{e}{-5}$ & Yes \\
	$1.5\mathrm{e}{-5}$ & $5 \mathrm{e}{-6}$ & $(M_x-1,\ M_c-1)$ & $6\mathrm{e}{-5}$ & $2\mathrm{e}{-5}$ & Yes \\ 
	\midrule[1.0pt]
	$8\mathrm{e}{-6}$ & $1.2\mathrm{e}{-5}$ & $(M_x,\ M_c)$ & $0.8\mathrm{e}{-5}$ & $1.2\mathrm{e}{-5}$ & No \\
	$8\mathrm{e}{-6}$ & $1.2\mathrm{e}{-5}$ & $(M_x-1,\ M_c)$ & $3.2\mathrm{e}{-5}$ & $1.2\mathrm{e}{-5}$ & Yes \\
	$8\mathrm{e}{-6}$ & $1.2 \mathrm{e}{-5}$ & $(M_x,\ M_c-1)$ & $0.8\mathrm{e}{-5}$ & $4.8\mathrm{e}{-5}$ & Yes \\
	$8\mathrm{e}{-6}$ & $1.2\mathrm{e}{-5}$ & $(M_x-1,\ M_c-1)$ & $3.2\mathrm{e}{-5}$ & $4.8\mathrm{e}{-5}$ & Yes \\
	
	\bottomrule
\end{tabular}

\label{table: Approximation Error for different parameters}
\end{table}

Interestingly, regardless of the tolerance constants $\varepsilon_x$ and $\varepsilon_c$ used, the maximum approximation error after removing one hyperplane is approximately a constant multiple of the maximum approximation error using all required hyperplanes. We now sketch the argument for such a condition~below.

For $p = 1, 2, \dots, M_x-1$, after removing the $p$th hyperplane, the maximum approximation error becomes $e_{p-1}(x^{\prime \prime})$, where $x^{\prime \prime}$ is the point such that $e_{p-1}(x^{\prime \prime}) = e_{p+1}(x^{\prime \prime})$. 
By Lemma~\ref{lemma: Maximum Approximation Errors} with the log-additively separable utility, we have
\begin{align}\label{eq: x_{p+2}}
x^{\prime \prime} = x^{\prime \prime}( x_{p+1} ):= \frac{ \log \left( \frac{1+ x_{p-1}}{1 + x_{p+1}} \right) + a_{p+1} x_{p+1} - a_{p-1} x_{p-1} }{a_{p+1} - a_{p-1}}.
\end{align}
Then, by Theorem~\ref{theorem: Successive Partition}, we recursively get 
$
x_{p+1} = (1 + \mathsf{a}_x)^{2} x_{p-1}+ (1 +  \mathsf{a}_x)  \mathsf{a}_x +  \mathsf{a}_x, 
$ 
with $ \mathsf{a}_x$ is defined in Theorem~\ref{theorem: Successive Partition}.
Substituting~\eqref{eq: x_{p+2}} into $e_{p+1}(x^{\prime \prime})$ and using the fact that $e_{p-1}(x^{\prime \prime}) = e_{p+1}(x^{\prime \prime})$, a lengthy but straightforward calculation leads to  that the relation between~$e_{p-1}(x^{\prime})$ and $e_{p+1}(x^{\prime \prime})$
\begin{align*}
\frac{e_{p+1}(x^{\prime \prime})}{e_{p-1}(x^{\prime})}  
&=\frac{e_{p-1}(x^{\prime \prime})}{e_{p-1}(x^{\prime})}
= \frac{\frac{2}{ \mathsf{a}_x + 2} u - \log \left( \frac{2}{ \mathsf{a}_x + 2} u \right) - 1}{\mu- \log \mu - 1},
\end{align*}
where $u = \frac{1 }{ \mathsf{a}_x} \log(1+  \mathsf{a}_x)$ and $\mu = \frac{(1 + \mathsf{a}_x )^2 }{\mathsf{a}_x} \log (1 + \mathsf{a}_x)$. Note that $ \mathsf{a}_x$ is a constant; hence, 
for $p = 1,2,\dots, M_x-1$, it is readily verified  that $\frac{e_{p+1}(x^{\prime \prime})}{e_{p-1}(x^{\prime})}$ is a positive constant. Similarly, for $q = 1,2,\dots, M_c-1$, an almost identical argument shows that~$\frac{e_{q+1}( c^{\prime \prime})}{e_{q-1}( c^{\prime})}$ is also a positive constant.

Therefore, by induction, we infer that for optimal numbers of hyperplanes~$M = M_x + M_c$, after removing the prior $v$ hyperplanes on $x$ and removing $w$ hyperplanes on $c$, where the removed hyperplanes are indexed as $l = 1, 2,\dots, M_x-1$ and $r = 1, 2,\dots, M_c-1$, the maximum approximation error becomes the linear combination of the individual approximation errors, i.e., $\kappa \varepsilon_x + \eta \varepsilon_c$ for some positive constants $\kappa$ and $\eta$.

\section{Concluding Remarks and Future Work}
In this paper, we presented an innovative approach to addressing the computational challenges inherent in large-scale robust portfolio optimization.
Specifically, we extended the supporting hyperplane approximation method to account for a general class of additively separable utilities and to incorporate market friction arising from turnover transaction costs. We developed a robust and efficient technique for solving  a class of distributionally robust portfolio problems using hyperplanes of return rate and transaction cost rate, which significantly generalize the work in \cite{hsieh2024solving}. Our approach is particularly effective for managing large asset sets and incorporates practical considerations such as portfolio rebalancing costs. We then applied this method to large-scale portfolio optimization using the constituent stocks of the S\&P 500 and a risk-free asset.

We showed that our extended hyperplane approximation method can achieve performance arbitrarily close to that of the original log-optimal portfolio while significantly reducing computational time. Specifically, our empirical studies showed that the required running times decreased from several thousand seconds to just a few seconds even when the turnover costs and sliding window implementation are involved.  
Furthermore, even without diversified holding constraints, incorporating the polyhedral ambiguity set of return distribution enables robust portfolio optimization. Setting the turnover cost limit also facilitates portfolio diversification.

In summary, our proposed method offers a robust, efficient, and scalable solution to large-scale robust portfolio optimization, addressing both theoretical and practical challenges. Future research may explore further refinements to the supporting hyperplane approximation and extend our approach to other types of ambiguity sets and utility functions. Additionally, incorporating temporal dependencies and dynamic correlations in the return distributions could provide a more comprehensive framework for portfolio optimization under uncertainty. More detailed directions are listed below.

\emph{Future Work.} 
It would be interesting to generalize the framework to a robust mixed-integer program involving cardinality or specific long/short constraints with various cost models. We envision that a modified hyperplane approximation approach can be developed while maintaining computational efficiency.

Additionally, considering the noisy nature of financial data, including missing values, outliers, and noise trader signals, another interesting direction would be to ensure that the considered ambiguity set actually covers the true return distribution; see also \cite{farokhi2023distributionally} for preliminary research in this direction. 
Lastly, the returns model of the paper was treated as an i.i.d. random sample from an unknown but ambiguous distribution. However, this might not fully capture the reality of financial markets, where returns often exhibit temporal dependencies and volatility clustering. Future work could explore relaxing this assumption.

 \medskip
\textbf{Acknowledgment.}
{
	This work was  partly supported by the National Science and Technology Council (NSTC), Taiwan, under grants NSTC112-2813-C-007-002-H and NSTC112-2221-E-007-078-.
}


\bibliographystyle{model5-names}
\bibliography{refs}

\begin{thebibliography}{42}
\expandafter\ifx\csname natexlab\endcsname\relax\def\natexlab#1{#1}\fi
\providecommand{\url}[1]{\texttt{#1}}
\providecommand{\href}[2]{#2}
\providecommand{\path}[1]{#1}
\providecommand{\DOIprefix}{doi:}
\providecommand{\ArXivprefix}{arXiv:}
\providecommand{\URLprefix}{URL: }
\providecommand{\Pubmedprefix}{pmid:}
\providecommand{\doi}[1]{\href{http://dx.doi.org/#1}{\path{#1}}}
\providecommand{\Pubmed}[1]{\href{pmid:#1}{\path{#1}}}
\providecommand{\bibinfo}[2]{#2}
\ifx\xfnm\relax \def\xfnm[#1]{\unskip,\space#1}\fi
\bibitem[{Beck(2023)}]{beck2014introduction}
\bibinfo{author}{Beck, A.} (\bibinfo{year}{2023}).
\newblock {\it \bibinfo{title}{{Introduction to Nonlinear Optimization: Theory,
  Algorithms, and Applications with Python and MATLAB}}\/}.
\newblock \bibinfo{publisher}{SIAM}.
\bibitem[{Bekjan(2004)}]{bekjan2004joint}
\bibinfo{author}{Bekjan, T.~N.} (\bibinfo{year}{2004}).
\newblock \bibinfo{title}{{On Joint Convexity of Trace Functions}}.
\newblock {\it \bibinfo{journal}{Linear Algebra and Its Applications}\/},  {\it
  \bibinfo{volume}{390}\/}, \bibinfo{pages}{321--327}.
\bibitem[{Bertsimas et~al.(2023)Bertsimas, Shtern \& Sturt}]{bertsimas2023data}
\bibinfo{author}{Bertsimas, D.}, \bibinfo{author}{Shtern, S.}, \&
  \bibinfo{author}{Sturt, B.} (\bibinfo{year}{2023}).
\newblock \bibinfo{title}{{A Data-Driven Approach to Multistage Stochastic
  Linear Optimization}}.
\newblock {\it \bibinfo{journal}{Management Science}\/},  {\it
  \bibinfo{volume}{69}\/}, \bibinfo{pages}{51--74}.
\bibitem[{Black \& Litterman(1992)}]{black1992global}
\bibinfo{author}{Black, F.}, \& \bibinfo{author}{Litterman, R.}
  (\bibinfo{year}{1992}).
\newblock \bibinfo{title}{{Global Portfolio Optimization}}.
\newblock {\it \bibinfo{journal}{Financial analysts journal}\/},  {\it
  \bibinfo{volume}{48}\/}, \bibinfo{pages}{28--43}.
\bibitem[{Blanchet et~al.(2022)Blanchet, Chen \&
  Zhou}]{blanchet2022distributionally}
\bibinfo{author}{Blanchet, J.}, \bibinfo{author}{Chen, L.}, \&
  \bibinfo{author}{Zhou, X.~Y.} (\bibinfo{year}{2022}).
\newblock \bibinfo{title}{{Distributionally Robust Mean-Variance Portfolio
  Selection with Wasserstein Distances}}.
\newblock {\it \bibinfo{journal}{Management Science}\/},  {\it
  \bibinfo{volume}{68}\/}, \bibinfo{pages}{6382--6410}.
\bibitem[{Blanchet \& Murthy(2019)}]{blanchet2019quantifying}
\bibinfo{author}{Blanchet, J.}, \& \bibinfo{author}{Murthy, K.}
  (\bibinfo{year}{2019}).
\newblock \bibinfo{title}{Quantifying distributional model risk via optimal
  rransport}.
\newblock {\it \bibinfo{journal}{Mathematics of Operations Research}\/},  {\it
  \bibinfo{volume}{44}\/}, \bibinfo{pages}{565--600}.
\bibitem[{Boyd \& Vandenberghe(2004)}]{boyd2004convex}
\bibinfo{author}{Boyd, S.}, \& \bibinfo{author}{Vandenberghe, L.}
  (\bibinfo{year}{2004}).
\newblock {\it \bibinfo{title}{{Convex Optimization}}\/}.
\newblock \bibinfo{publisher}{Cambridge University Press}.
\bibitem[{Choi et~al.(2016)Choi, Rujeerapaiboon \& Jiang}]{choi2016multi}
\bibinfo{author}{Choi, B.-G.}, \bibinfo{author}{Rujeerapaiboon, N.}, \&
  \bibinfo{author}{Jiang, R.} (\bibinfo{year}{2016}).
\newblock \bibinfo{title}{{Multi-Period Portfolio Optimization: Translation of
  Autocorrelation Risk to Excess Variance}}.
\newblock {\it \bibinfo{journal}{Operations Research Letters}\/},  {\it
  \bibinfo{volume}{44}\/}, \bibinfo{pages}{801--807}.
\bibitem[{Cover \& Thomas(2012)}]{Cover_Thomas_2012}
\bibinfo{author}{Cover, T.~M.}, \& \bibinfo{author}{Thomas, J.~A.}
  (\bibinfo{year}{2012}).
\newblock {\it \bibinfo{title}{{Elements of Information Theory}}\/}.
\newblock \bibinfo{publisher}{Wiley}.
\bibitem[{Delage \& Ye(2010)}]{delage2010distributionally}
\bibinfo{author}{Delage, E.}, \& \bibinfo{author}{Ye, Y.}
  (\bibinfo{year}{2010}).
\newblock \bibinfo{title}{{Distributionally Robust Optimization under Moment
  Uncertainty with Application to Data-Driven Problems}}.
\newblock {\it \bibinfo{journal}{Operations Research}\/},  {\it
  \bibinfo{volume}{58}\/}, \bibinfo{pages}{595--612}.
\bibitem[{Duffie \& Pan(1997)}]{duffie1997overview}
\bibinfo{author}{Duffie, D.}, \& \bibinfo{author}{Pan, J.}
  (\bibinfo{year}{1997}).
\newblock \bibinfo{title}{{An Overview of Value at Risk}}.
\newblock {\it \bibinfo{journal}{Journal of Derivatives}\/},  {\it
  \bibinfo{volume}{4}\/}, \bibinfo{pages}{7--49}.
\bibitem[{Fabozzi et~al.(2007)Fabozzi, Kolm, Pachamanova \&
  Focardi}]{fabozzi2007robust}
\bibinfo{author}{Fabozzi, F.~J.}, \bibinfo{author}{Kolm, P.~N.},
  \bibinfo{author}{Pachamanova, D.~A.}, \& \bibinfo{author}{Focardi, S.~M.}
  (\bibinfo{year}{2007}).
\newblock {\it \bibinfo{title}{{Robust Portfolio Optimization and
  Management}}\/}.
\newblock \bibinfo{publisher}{John Wiley \& Sons}.
\bibitem[{Farokhi(2023)}]{farokhi2023distributionally}
\bibinfo{author}{Farokhi, F.} (\bibinfo{year}{2023}).
\newblock \bibinfo{title}{{Distributionally-Robust Optimization with Noisy Data
  for Discrete Uncertainties Using Total Variation Distance}}.
\newblock {\it \bibinfo{journal}{IEEE Control Systems Letters}\/}, .
\bibitem[{Feng et~al.(2016)Feng, Palomar et~al.}]{feng2016signal}
\bibinfo{author}{Feng, Y.}, \bibinfo{author}{Palomar, D.~P.} et~al.
  (\bibinfo{year}{2016}).
\newblock \bibinfo{title}{{A Signal Processing Perspective on Financial
  Engineering}}.
\newblock {\it \bibinfo{journal}{Foundations and Trends{\textregistered} in
  Signal Processing}\/},  {\it \bibinfo{volume}{9}\/}, \bibinfo{pages}{1--231}.
\bibitem[{Ghahtarani et~al.(2022)Ghahtarani, Saif \&
  Ghasemi}]{ghahtarani2022robust}
\bibinfo{author}{Ghahtarani, A.}, \bibinfo{author}{Saif, A.}, \&
  \bibinfo{author}{Ghasemi, A.} (\bibinfo{year}{2022}).
\newblock \bibinfo{title}{{Robust Portfolio Selection Problems: a Comprehensive
  Review}}.
\newblock {\it \bibinfo{journal}{Operational Research}\/},  {\it
  \bibinfo{volume}{22}\/}, \bibinfo{pages}{3203--3264}.
\bibitem[{Hsieh(2023)}]{hsieh2023asymptotic}
\bibinfo{author}{Hsieh, C.-H.} (\bibinfo{year}{2023}).
\newblock \bibinfo{title}{{On Asymptotic Log-Optimal Portfolio Optimization}}.
\newblock {\it \bibinfo{journal}{Automatica}\/},  {\it
  \bibinfo{volume}{151}\/}, \bibinfo{pages}{110901}.
\bibitem[{Hsieh(2024)}]{hsieh2024solving}
\bibinfo{author}{Hsieh, C.-H.} (\bibinfo{year}{2024}).
\newblock \bibinfo{title}{{On Solving Robust Log-Optimal Portfolio: A
  Supporting Hyperplane Approximation Approach}}.
\newblock {\it \bibinfo{journal}{European Journal of Operational Research}\/},
  {\it \bibinfo{volume}{313}\/}, \bibinfo{pages}{1129--1139}.
\bibitem[{Jorion(2007)}]{jorion2007value}
\bibinfo{author}{Jorion, P.} (\bibinfo{year}{2007}).
\newblock {\it \bibinfo{title}{{Value at Risk: the New Benchmark for Managing
  Financial Risk}}\/}.
\newblock \bibinfo{publisher}{McGraw-Hill}.
\bibitem[{Kelly(1956)}]{kelly1956new}
\bibinfo{author}{Kelly, J.~L.} (\bibinfo{year}{1956}).
\newblock \bibinfo{title}{{A New Interpretation of Information Rate}}.
\newblock {\it \bibinfo{journal}{the Bell System Technical Journal}\/},  {\it
  \bibinfo{volume}{35}\/}, \bibinfo{pages}{917--926}.
\bibitem[{Konno \& Yamazaki(1991)}]{konno1991mean}
\bibinfo{author}{Konno, H.}, \& \bibinfo{author}{Yamazaki, H.}
  (\bibinfo{year}{1991}).
\newblock \bibinfo{title}{{Mean-Absolute Deviation Portfolio Optimization Model
  and its Applications to Tokyo Stock Market}}.
\newblock {\it \bibinfo{journal}{Management Science}\/},  {\it
  \bibinfo{volume}{37}\/}, \bibinfo{pages}{519--531}.
\bibitem[{Lee et~al.(2023)Lee, Bae, Kim \& Lee}]{lee2023value}
\bibinfo{author}{Lee, J.}, \bibinfo{author}{Bae, S.}, \bibinfo{author}{Kim,
  W.~C.}, \& \bibinfo{author}{Lee, Y.} (\bibinfo{year}{2023}).
\newblock \bibinfo{title}{{Value Function Gradient Learning for Large-Scale
  Multistage Stochastic Programming Problems}}.
\newblock {\it \bibinfo{journal}{European Journal of Operational Research}\/},
  {\it \bibinfo{volume}{308}\/}, \bibinfo{pages}{321--335}.
\bibitem[{Li et~al.(2018)Li, Wang, Huang \& Hoi}]{li2018transaction}
\bibinfo{author}{Li, B.}, \bibinfo{author}{Wang, J.}, \bibinfo{author}{Huang,
  D.}, \& \bibinfo{author}{Hoi, S.~C.} (\bibinfo{year}{2018}).
\newblock \bibinfo{title}{{Transaction Cost Optimization for Online Portfolio
  Selection}}.
\newblock {\it \bibinfo{journal}{Quantitative Finance}\/},  {\it
  \bibinfo{volume}{18}\/}, \bibinfo{pages}{1411--1424}.
\bibitem[{Li(2023)}]{li2023wasserstein}
\bibinfo{author}{Li, J. Y.-M.} (\bibinfo{year}{2023}).
\newblock \bibinfo{title}{{Wasserstein-Kelly Portfolios: A Robust Data-Driven
  Solution to Optimize Portfolio Growth}}.
\newblock {\it \bibinfo{journal}{arXiv preprint arXiv:2302.13979}\/}, .
\bibitem[{Luenberger(1993)}]{luenberger1993preference}
\bibinfo{author}{Luenberger, D.~G.} (\bibinfo{year}{1993}).
\newblock \bibinfo{title}{{A Preference Foundation for Log Mean-Variance
  Criteria in Portfolio Choice Problems}}.
\newblock {\it \bibinfo{journal}{Journal of Economic Dynamics and Control}\/},
  {\it \bibinfo{volume}{17}\/}, \bibinfo{pages}{887--906}.
\bibitem[{Luenberger(2013)}]{luenberger2013investment}
\bibinfo{author}{Luenberger, D.~G.} (\bibinfo{year}{2013}).
\newblock {\it \bibinfo{title}{{Investment Science}}\/}.
\newblock \bibinfo{publisher}{Oxford University Press}.
\bibitem[{MacLean et~al.(2011)MacLean, Thorp \&
  Ziemba}]{MacLean_Thorp_Ziemba_2011book}
\bibinfo{author}{MacLean, L.~C.}, \bibinfo{author}{Thorp, E.~O.}, \&
  \bibinfo{author}{Ziemba, W.~T.} (\bibinfo{year}{2011}).
\newblock {\it \bibinfo{title}{{The Kelly Capital Growth Investment Criterion:
  Theory and Practice}}\/} volume~\bibinfo{volume}{3}.
\newblock \bibinfo{publisher}{World Scientific}.
\bibitem[{Markowitz(1952)}]{markowitz1952}
\bibinfo{author}{Markowitz, H.} (\bibinfo{year}{1952}).
\newblock \bibinfo{title}{{Portfolio Selection}}.
\newblock {\it \bibinfo{journal}{The Journal of Finance}\/},  {\it
  \bibinfo{volume}{7}\/}, \bibinfo{pages}{77--91}.
\bibitem[{Michaud(1989)}]{michaud1989markowitz}
\bibinfo{author}{Michaud, R.~O.} (\bibinfo{year}{1989}).
\newblock \bibinfo{title}{{The Markowitz Optimization Enigma: Is ``Optimized'''
  Optimal?}}
\newblock {\it \bibinfo{journal}{Financial Analysts Journal}\/},  {\it
  \bibinfo{volume}{45}\/}, \bibinfo{pages}{31--42}.
\bibitem[{Mohajerin~Esfahani \& Kuhn(2018)}]{mohajerin2018data}
\bibinfo{author}{Mohajerin~Esfahani, P.}, \& \bibinfo{author}{Kuhn, D.}
  (\bibinfo{year}{2018}).
\newblock \bibinfo{title}{{Data-Driven Distributionally Robust Optimization
  using the Wasserstein Metric: Performance Guarantees and Tractable
  Reformulations}}.
\newblock {\it \bibinfo{journal}{Mathematical Programming}\/},  {\it
  \bibinfo{volume}{171}\/}, \bibinfo{pages}{115--166}.
\bibitem[{Newbury(1981)}]{Newbury1981}
\bibinfo{author}{Newbury, J.} (\bibinfo{year}{1981}).
\newblock \bibinfo{title}{{{Linear Interpolation}}}.
\newblock In {\it \bibinfo{booktitle}{Basic Numeracy Skills and Practice}\/}
  (pp. \bibinfo{pages}{67--72}).
\newblock \bibinfo{address}{London}: \bibinfo{publisher}{Macmillan Education
  UK}.
\bibitem[{Perold(1984)}]{perold1984large}
\bibinfo{author}{Perold, A.~F.} (\bibinfo{year}{1984}).
\newblock \bibinfo{title}{{Large-Scale Portfolio Optimization}}.
\newblock {\it \bibinfo{journal}{Management Science}\/},  {\it
  \bibinfo{volume}{30}\/}, \bibinfo{pages}{1143--1160}.
\bibitem[{Potaptchik et~al.(2008)Potaptchik, Tun{\c{c}}el \&
  Wolkowicz}]{potaptchik2008large}
\bibinfo{author}{Potaptchik, M.}, \bibinfo{author}{Tun{\c{c}}el, L.}, \&
  \bibinfo{author}{Wolkowicz, H.} (\bibinfo{year}{2008}).
\newblock \bibinfo{title}{{Large Scale Portfolio Optimization with Piecewise
  Linear Transaction Costs}}.
\newblock {\it \bibinfo{journal}{Optimization Methods \& Software}\/},  {\it
  \bibinfo{volume}{23}\/}, \bibinfo{pages}{929--952}.
\bibitem[{Rahimian \& Mehrotra(2019)}]{rahimian2019distributionally}
\bibinfo{author}{Rahimian, H.}, \& \bibinfo{author}{Mehrotra, S.}
  (\bibinfo{year}{2019}).
\newblock \bibinfo{title}{{Distributionally Robust Optimization: A Review}}.
\newblock {\it \bibinfo{journal}{arXiv preprint arXiv:1908.05659}\/}, .
\bibitem[{Rockafellar et~al.(2000)Rockafellar, Uryasev
  et~al.}]{rockafellar2000optimization}
\bibinfo{author}{Rockafellar, R.~T.}, \bibinfo{author}{Uryasev, S.} et~al.
  (\bibinfo{year}{2000}).
\newblock \bibinfo{title}{{Optimization of Conditional Value-at-Risk}}.
\newblock {\it \bibinfo{journal}{Journal of Risk}\/},  {\it
  \bibinfo{volume}{2}\/}, \bibinfo{pages}{21--42}.
\bibitem[{Rujeerapaiboon et~al.(2016)Rujeerapaiboon, Kuhn \&
  Wiesemann}]{rujeerapaiboon2016robust}
\bibinfo{author}{Rujeerapaiboon, N.}, \bibinfo{author}{Kuhn, D.}, \&
  \bibinfo{author}{Wiesemann, W.} (\bibinfo{year}{2016}).
\newblock \bibinfo{title}{{Robust Growth-Optimal Portfolios}}.
\newblock {\it \bibinfo{journal}{Management Science}\/},  {\it
  \bibinfo{volume}{62}\/}, \bibinfo{pages}{2090--2109}.
\bibitem[{Ryoo(2007)}]{ryoo2007compact}
\bibinfo{author}{Ryoo, H.~S.} (\bibinfo{year}{2007}).
\newblock \bibinfo{title}{{A Compact Mean-Variance-Skewness Model for
  Large-Scale Portfolio Optimization and Its Application to the NYSE Market}}.
\newblock {\it \bibinfo{journal}{Journal of the Operational Research
  Society}\/},  {\it \bibinfo{volume}{58}\/}, \bibinfo{pages}{505--515}.
\bibitem[{Shapiro et~al.(2021)Shapiro, Dentcheva \&
  Ruszczynski}]{shapiro2021lectures}
\bibinfo{author}{Shapiro, A.}, \bibinfo{author}{Dentcheva, D.}, \&
  \bibinfo{author}{Ruszczynski, A.} (\bibinfo{year}{2021}).
\newblock {\it \bibinfo{title}{{Lectures on Stochastic Programming: Modeling
  and Theory}}\/}.
\newblock \bibinfo{publisher}{SIAM}.
\bibitem[{Steinbach(2001)}]{steinbach2001markowitz}
\bibinfo{author}{Steinbach, M.~C.} (\bibinfo{year}{2001}).
\newblock \bibinfo{title}{{Markowitz Revisited: Mean-Variance Models in
  Financial Portfolio Analysis}}.
\newblock {\it \bibinfo{journal}{SIAM Review}\/},  {\it
  \bibinfo{volume}{43}\/}, \bibinfo{pages}{31--85}.
\bibitem[{Sun \& Boyd(2018)}]{sun2018distributional}
\bibinfo{author}{Sun, Q.}, \& \bibinfo{author}{Boyd, S.}
  (\bibinfo{year}{2018}).
\newblock \bibinfo{title}{{Distributional Robust Kelly Gambling: Optimal
  Strategy under Uncertainty in the Long-Run}}.
\newblock {\it \bibinfo{journal}{arXiv preprint arXiv:1812.10371}\/}, .
\bibitem[{Takehara(1993)}]{takehara1993interior}
\bibinfo{author}{Takehara, H.} (\bibinfo{year}{1993}).
\newblock \bibinfo{title}{{An Interior Point Algorithm for Large Scale
  Portfolio Optimization}}.
\newblock {\it \bibinfo{journal}{Annals of Operations Research}\/},  {\it
  \bibinfo{volume}{45}\/}, \bibinfo{pages}{373--386}.
\bibitem[{Wang \& Hsieh(2022)}]{wang2022data}
\bibinfo{author}{Wang, P.-T.}, \& \bibinfo{author}{Hsieh, C.-H.}
  (\bibinfo{year}{2022}).
\newblock \bibinfo{title}{{On Data-Driven Log-Optimal Portfolio: A Sliding
  Window Approach}}.
\newblock {\it \bibinfo{journal}{IFAC-PapersOnLine}\/},  {\it
  \bibinfo{volume}{55}\/}, \bibinfo{pages}{474--479}.
\bibitem[{{Yahoo! Finance}(2024)}]{YahooFinance}
\bibinfo{author}{{Yahoo! Finance}} (\bibinfo{year}{2024}).
\newblock \bibinfo{title}{Yahoo! finance}.
\newblock \bibinfo{howpublished}{\url{https://finance.yahoo.com}}.
\newblock \bibinfo{note}{Accessed: March 01, 2024.}

\end{thebibliography}


\appendix

\section{Technical Proofs}


\subsection{Proofs in Section~\ref{section: Preliminaries}} \label{appendix: proofs in preliminaries}

\begin{proof}[Proof of Lemma~\ref{lemma: Weak Survivability Condition}] 
Since $V(0) > 0$, we consider $t > 0$.
Assume $V(t - 1) \geq 0$. Then
$
\sum_{i=1}^{n} K_i^{+}(t) |\min\{0, x_{i, \min}\}|- \sum_{i=1}^{n} K_i^{-}(t) \max\{0, x_{i, \max}\} \leq 1
$
is equivalent to 
$
\sum_{i=1}^{n} K_i^{+}(t) \min\{0, x_{i, \min}\} + \sum_{i=1}^{n} K_i^{-}(t) \max\{0, x_{i, \max}\} \geq -1.
$
Moreover, since $-1<x_{i, \min} \leq x_i(t) \leq x_{i, \max}$, we have that
{\small	\begin{align*}
	\sum_{i=1}^{n} K_i^{+}(t) x_i(t) + \sum_{i=1}^{n} K_i^{-}(t) x_i(t)  
	& \geq \sum_{i=1}^{n} K_i^{+}(t) \min\{0, x_{i, \min}\}+ \sum_{i=1}^{n} K_i^{-}(t) \max\{0, x_{i, \max}\}\\ & \geq -1.
\end{align*}
}	Then $K(t)^\top X(t) \geq -1$ and $1 + K(t)^\top X(t)\geq 0$. Moreover, by the assumed hypothesis, we have
$ 1 - |K(t) - K(t-1))|^\top C(t) \geq 0 $, Therefore, the account value at $t$ is
$
V(t) = (1 + K(t)^\top  X(t))(1 - |K(t) - K(t-1))|^\top C(t)) V(t-1) \geq 0.
$
By induction, the required survival constraints $V(t)\geq0$ hold.
\end{proof}


\begin{proof}[Proof of Theorem~\ref{theorem: An Equivalent Distributional Robust Optimization Problem}]
By Lemma~\ref{lemma: data-driven expression of the running ELG}, we have
\begin{align*}
	J_p(t; K(t), K(t-1)) 
	&= \sum_{j=1}^m p_j  U_t \left( ( 1 + K(t)^\top x^j)(1 - | K(t) - K(t-1))|^\top C(t)) \right).
\end{align*}
Define $
q(K(t)) = [q(K(t))_1\;  \cdots \; q(K(t))_m]^\top
$
with the $j$th component satisfying~\eqref{eq: q_K_t_j}.
Consider the Lagrangian
$$
\mathcal{L}(\nu, \lambda, p) = p^\top q( K(t) ) +\nu^\top (A_0p - d_0) + \lambda^\top(A_1 p - d_1)+ \mathbbm{1}_{S_m}(p)
$$
where the Lagrange multipliers are $\nu \in \mathbb{R}^{m_0}
$, $\lambda \in \mathbb{R}^{m_1}$ with $\lambda_j \geq 0$,  
and the indicator function $\mathbbm{1}_{S_m}$ represents the probability simplex condition that for $p \in \mathbb{R}^m_+$, 
$$
	\mathbbm{1}_{S_m}( p ) :=
	\begin{cases}
		1,  & p \in S_m \\
		0, & {\rm otherwise}.
	\end{cases}
$$
Then, the Lagrangian dual function is
\begin{align*}
	h( \nu, \lambda )
	&= \inf_{p \in S_m} \mathcal{L}(\nu, \lambda, p) \\
	&= \min_{p\in S_m}p^\top (q(K(t)) +A_0^\top \nu + A_1^\top \lambda) - \nu^\top d_0 - \lambda^\top d_1 \\
	&= \min_{p \in S_m} \sum_{ j=1 }^m p_j(q(K(t)) + A_0^\top\nu + A_1^\top \lambda)_j - \nu^\top d_0 - \lambda^\top d_1\\
	&\geq \min_{j} (q(K(t)) + A_0^\top \nu + A_1^\top \lambda)_j - \nu^\top d_0 - \lambda^\top d_1
\end{align*}
where the last inequality holds since 
$$ 
(q(K(t)) + A_0^\top\nu + A_1^\top \lambda)_j \geq \mathop{\min}_{j} ( q(K(t)) + A_0^\top\nu + A_1^\top \lambda)_j.
$$ 
Hence, the dual problem of $\mathop{\inf}_{p \in \mathcal{P}} g_p(t)$ is
\begin{align*}
	&\max_{\nu, \lambda} \min_{j} \; (q(K(t)) + A_0^\top\nu + A_1^\top \lambda)_j - \nu^\top d_0 - \lambda^\top d_1 \\
	&\text{s.t.}\; \lambda \succeq 0  
\end{align*}
which has the same optimal value as $\mathop{\inf}_{p \in \mathcal{P}} g_p(t)$ since the strong duality holds by Slater's condition that $p \in \mathcal{P}$ and $A_1p = d_1$ are affine. Then, the distributional robust log-optimal portfolio problem can be written~as
\begin{align*}
	&\mathop{\max}_{K(t), \nu, \lambda} \mathop{\min}_{j} ( q(K(t) ) + A_0^\top \nu+ A_1^\top \lambda)_j - \nu^\top d_0 - \lambda^\top d_1 \\
	& \text{s.t.} \; K(t) \in \mathcal{K}, \lambda \succeq 0. \qedhere
\end{align*} 
\end{proof}

\subsection{Proofs in Section~\ref{section: Extended Supporting Hyperplane Approximation}}  \label{appendix: proofs in extended supporting hyperplane approximaiton}

\medskip
\begin{proof}[Proof of Lemma~\ref{lemma: Limit Behavior And Monotonicity of Approximate Error}]
To prove part~$(i)$, we observe that for any $l = 0, 1, \dots, M_x$ and $x \in [x_{\min}, x_{\max}]$, the approximation error $e_l(x) = a_l(x - x_l)+ \alpha_t \cdot \phi_{1,t}(x_l)   - \alpha_t \cdot \phi_{1,t}(x)$. By taking derivative of $e_l(x)$, we obtain  
\begin{align}
	\frac{d }{d x} e_l(x)
 = a_l - \alpha_t  \phi_{1,t}'(x) 
	& = \alpha_t  \phi_{1,t}'(x_l) - \alpha_t  \phi_{1,t}'(x)  \label{eq: F_t eq}
\end{align}
Given that $\phi_{1,t}$ is strictly concave and strictly increasing, $\phi_{1,t}'(x) > 0$ and $\phi_{1,t}'(x)$ is strictly decreasing. Therefore, for $x \in (x_l, x_{\max}]$, we have 
$
\phi_{1,t}'(x) < \phi_{1,t}'(x_l).
$
Hence, \eqref{eq: F_t eq} becomes
$
\frac{d }{d x} e_l(x) >  \alpha_t  \phi_{1,t}'(x) - \alpha_t  \phi_{1,t}'(x) = 0.
$

On the other hand, for $x \in [x_{\min}, x_l)$, given that $\phi_{1,t}$ is strictly concave and strictly increasing, $\phi_{1,t}'(x) > 0$ and $\phi_{1,t}'(x)$ is strictly decreasing. Therefore, for $x \in [x_{\min}, x_l)$, we have 
$
\phi_{1,t}'(x) > \phi_{1,t}'(x_l).
$
Hence, \eqref{eq: F_t eq} becomes
$
\frac{d }{d x} e_l(x) < \alpha_t \cdot \phi_{1,t}'(x) - \alpha_t \cdot \phi_{1,t}'(x) =0.
$
Lastly, as \( x \to x_l \), we have
\( \phi_{1,t}(x) \to \phi_{1,t}(x_l) \). 
Therefore,
\[
\lim_{x \to x_l} e_l(x) = a_l (x_l - x_l) + \alpha_t \cdot \phi_{1,t}(x_l) - \alpha_t \cdot \phi_{1,t}(x_l) = 0.
\]

To prove part~$(ii)$, an almost identical proof as part~$(i)$ would work. Specifically, for any $r=0,1,\dots,M_c$ and $c \in [c_{\min}, c_{\max}]$, recall that $e_r(c) := b_r (c - c_r) + \beta_t \cdot \phi_{2,t}(c_r)  - \beta_t \cdot \phi_{2,t}(c)$, the  derivative of  $e_r(c)$ is given by
\begin{align}
	\frac{d }{d c} e_r(c) 
	= b_r - \beta_t  \phi_{2,t}'(c)  
	&= \beta_t  \phi_{2,t}'(c_r) - \beta_t  \phi_{2,t}'(c)  \label{eq: G_t eq}
\end{align}
Given that $\phi_{2,t}$ is strictly concave and strictly decreasing, $\phi_{2,t}'(c)< 0$ and $\phi_{2,t}'$ is strictly decreasing. Therefore, for $c \in (c_r, c_{\max}]$, we have
$
\phi_{2,t}'(c_r) > \phi_{2,t}'(c).
$
Hence, \eqref{eq: G_t eq} becomes
$
	\frac{d }{d c} e_r(c) 
	> \beta_t \cdot \phi_{2,t}'(c) - \beta_t \cdot \phi_{2,t}'(c) = 0.
$
On the other hand, for $c \in [c_{\min}, c_r)$, given that $\phi_{2, t}$ is convex and strictly increasing, $\phi_{2,t}'(x) > 0$ and $\phi_{2,t}'(x)$ is increasing. Therefore, for $c \in [c_{\min}, c_r)$, we have 
$
\phi_{2,t}'(c) > \phi_{2,t}'(c_r).
$
Hence, \eqref{eq: G_t eq} becomes
\begin{align}
	\frac{d }{d c} e_r(c) 
	&<   \beta_t \cdot \phi_{2,t}'(c) - \beta_t \cdot \phi_{2,t}'(c) = 0.
\end{align}
Lastly, as \( c \to c_r \), we have
\( \phi_{2, t}(c) \to \phi_{2,t}(c_r) \). 
Therefore,
\[
\lim_{c \to c_r} e_r(c) = b_r (c_r - c_r) + \beta_t \cdot \phi_{2, t}(c_r) - \beta_t \cdot \phi_{2, t}(c_r) = 0. \qedhere
\]
\end{proof}

\medskip
\begin{proof}[Proof of Theorem~\ref{theorem: Separable Maximum Approximation Error}]  
We begin by observing that
{\small	\begin{align}
	&\sup_{x,c} \left| f(x, c) - \min_{l, r}  h_{l, r}(x, c)  \right| \nonumber \\
	&\quad = \sup_{x, c} \left| U\left( ( 1 + x )(1 -c) \right) - \min_{l, r} \left\{ a_l x + b_r c + \gamma_{l, r} \right\} \right| \nonumber \\
	&\quad = \sup_{x,c} \left| U\left( ( 1 + x )(1 -c) \right) - \min_{l, r} \left\{ a_l x + b_r c + U((1+x_l) (1- c_r)) - a_l x_l - b_r c_r \right\} \right| \nonumber \\
	&\quad = \sup_{x,c} \left| \alpha_t  \phi_{1,t}(x) + \beta_t  \phi_{2,t}(c) - \min_{l, r} \left\{ a_l x + b_r c + \alpha_t  \phi_{1,t}(x_l) + \beta_t   \phi_{2,t}(c_r) - a_l x_l - b_r c_r \right\} \right| \nonumber.
\end{align}
}where the last equality holds by the additive separability that $U((1+x)(1-c)) = \alpha_t \cdot \phi_{1,t}(x) + \beta_t \cdot \phi_{2,t}(c)$, we have
{\small	\begin{align}
	&\sup_{x,c} \left| f(x, c) - \min_{l, r}  h_{l, r}(x, c)  \right| \nonumber \\
	&\quad = \sup_{x,c} \left| \alpha_t  \phi_{1,t}(x) + \beta_t  \phi_{2,t}(c)  - \min_{l} \{ a_l x + \alpha_t  \phi_{1,t}(x_l)  -a_lx_l \} - \min_{r} \{b_r c + \beta_t  \phi_{2,t}(c_r) - b_rc_r\} \right| \nonumber \\
	&\quad = \sup_{x,c} \left|  \min_{l} \left\{ a_l (x - x_l) + \alpha_t  \phi_{1,t}(x_l) - \alpha_t  \phi_{1,t}(x)  \right\} + \min_{r} \left\{ b_r (c - c_r) + \beta_t  \phi_{2,t}(c_r)  - \beta_t  \phi_{2,t}(c) \right\} \right| \nonumber \\
	&\quad = \sup_{x,c} \left|  \min_l e_l(x) + \min_r e_r(c) \right|, \label{eq: approximation error --intermediate step}
\end{align}		
}	where $e_l(x)$ and $e_r(c)$ are defined in~\eqref{eq: approximation error along x-direction} and~\eqref{eq: approximation error along c-direction}.

By Lemma~\ref{lemma: Limit Behavior And Monotonicity of Approximate Error}, since $e_l(x)$ is strictly decreasing in $[x_{\min}, x_l)$ and $e_l(x)$ is strictly increasing in~$(x_l, x_{\max}]$, and by the fact that $e_l(x)=0$ at $x_l$, we obtain that $e_l(x) \geq 0$ for $l=0,1,\dots,M_x$. In the same way, we have $e_r(c) \geq 0$ for $r = 0, 1, \dots, M_c$. 
This implies that 
$
\min_l e_l(x) \geq 0 \text{ and }  \min_r e_r(c) \geq 0.
$
Therefore, Equation~\eqref{eq: approximation error --intermediate step} becomes
\begin{align*}
	\sup_{x,c} \left|  \min_l e_l(x) + \min_r e_r(c) \right|
	& = \sup_{x,c} \left\{ \min_l e_l(x) + \min_r e_r(c) \right\}\\
	& = \sup_{x} \min_{l} e_l(x) + \sup_{c} \min_{r} e_r(c)
\end{align*}
Hence, the proof is complete.
\end{proof}

\begin{proof}[Proof of Corollary~\ref{corollary: Separable Maximum Approximation Error for Partitions}]
By the proof of Theorem~\ref{theorem: Separable Maximum Approximation Error}, since $e_l(x) \geq 0$ for $x \in [x_p, x_{p+1}] \subseteq [x_{\min}, x_{\max}]$, where $p=0,1,\dots,M_x-1$, and  $e_r(c) \geq 0$ for $c \in [c_q, c_{q+1}] \subseteq [c_{\min}, c_{\max}]$, where $q = 0, 1,\dots, M_c-1$, we have
\begin{align*}
	\sup_{ \substack{ x \in [x_p, x_{p+1}]\\ c \in [c_q, c_{q+1}]} } \left| f(x, c) - \min_{l, r} h_{l, r}(x, c)  \right| 
	& =   \sup_{ \substack{ x \in [x_p, x_{p+1}] \\ c \in [c_q, c_{q+1}]}} \left| \min_l e_l(x) + \min_r e_r(c)   \right| \\
	& = \sup_{x \in [x_p, x_{p+1}]} \min_{l} e_l(x) + \sup_{c \in [c_q, c_{q+1}]} \min_{r} e_r(c),
\end{align*}
which is desired.
\end{proof}

\begin{proof}[Proof of Lemma~\ref{lemma: Maximum Approximation Errors}] 
We begin by considering the partition $\{x_l\}_{l=0}^{M_x}$ such that $x_0 = x_{\min}$ and $x_{M_x} = x_{\max}$. 
Fix $p \in \{0, 1, \dots, M_x-1\}$.
According to Corollary~\ref{corollary: Separable Maximum Approximation Error for Partitions}, the maximum approximation error is separable; i.e.,
\[
\sup_{ \substack{ x \in [x_p, x_{p+1}],\\ c \in [c_q, c_{q+1}]} } | f(x, c) - h_{l, r}(x, c) | = \sup_{x \in [x_p, x_{p+1}]} \min_{l} e_l(x) + \sup_{c \in [c_q, c_{q+1}]} \min_{r} e_r(c).
\]
By part~$(i)$ of Lemma~\ref{lemma: Limit Behavior And Monotonicity of Approximate Error}, 
it follows that the error $e_p(x)$ is strictly increasing in $(x_p, x_{\max}]$ and~$e_{p+1}(x)$ is strictly decreasing in $[x_{\min}, x_{p+1})$. Moreover, since $\{x_l\}$ are partition points, we have~$x_{p+1} > x_{p}$. Therefore, it implies that there exists $x^{\prime} \in (x_p, x_{p+1})$ such that $e_p(x^{\prime}) = e_{p+1}(x^{\prime})$. 
Then we now show that for such $x^{\prime}$, 
$$
\sup_{x \in [x_p, x_{p+1}]} \min_{l=0,1,\dots,M_x} e_l(x) = e_p(x^{\prime}) = e_{p+1}(x^{\prime}).
$$

Again, since $e_p(x)$ is strictly increasing in $(x_p, x_{p+1}] \subseteq (x_p, x_{\max}]$ and $e_{p+1}(x)$ is strictly decreasing in $[x_p, x_{p+1}) \subseteq [x_{\min}, x_{p+1})$, and with the fact that $e_p(x_p)=0$ and $e_{p+1}(x_{p+1})=0$, we obtain two cases: 

\textit{Case 1.} For  $x \in [x_p, x^{\prime}]$, we have 
\begin{align} \label{ineq: e_p leq e_{p+1}}
	e_p(x) \leq e_{p+1}(x).
\end{align}

\textit{Case 2.} For $x \in (x^{\prime}, x_{p+1}]$, we have 
\begin{align} \label{ineq: e_p > e_{p+1}}
	e_p(x) > e_{p+1}(x).
\end{align}

Note that for $l \leq p$, the difference $e_l(x) - e_p(x) \geq 0$ for  $x \in [x_p, x_{p+1}]$. To see this, we note that
\begin{align}
	e_l(x) - e_p(x) 
	&= 		a_l (x - x_l) + \alpha_t \cdot \phi_{1,t}(x_l) - \alpha_t \cdot \phi_{1,t}(x) \notag \\
	& \qquad - [ a_p (x - x_p) + \alpha_t \cdot \phi_{1,t}(x_p) - \alpha_t \cdot \phi_{1,t}(x)  ] \notag \\
	& = \alpha_t [\phi_{1,t}'(x_l)(x - x_l) -   \phi_{1,t}'(x_p)(x - x_p) +   \phi_{1,t}(x_l)  -   \phi_{1,t}(x_p) ].
\end{align}
Note that $\phi_{1,t}$ is strictly concave, $-\phi_{1,t}$ is strictly convex. Hence, it has a first-order lower bound, see \cite{beck2014introduction}, i.e.,
$
-\phi_{1,t}(x_p) \geq -\phi_{1,t}(x_l)  - \phi_{1,t}'(x_l)(x_p - x_l)
$
which leads to
\begin{align}
&	e_l(x) - e_p(x) \notag \\
	& \geq	 \alpha_t [\phi_{1,t}'(x_l)(x - x_l) -   \phi_{1,t}'(x_p)(x - x_p) +   \phi_{1,t}(x_l)  -\phi_{1,t}(x_l)  - \phi_{1,t}'(x_l)(x_p - x_l) ]  \notag \\
	& = \alpha_t [\phi_{1,t}'(x_l)(x - x_l) -   \phi_{1,t}'(x_p)(x - x_p)  - \phi_{1,t}'(x_l)(x_p - x_l) ]  \notag \\
	& = \alpha_t [\phi_{1,t}'(x_l)-   \phi_{1,t}'(x_p)]  (x - x_p)  \geq 0 \label{ineq: e_l - e_p}
\end{align}
where the last inequality holds since $\phi_{1,t}$ is strictly concave and increasing, it implies that $\phi_{1,t}'>0$ and $\phi_{1,t}'$ is strictly decreasing; that is, for $l \leq p$, it follows that $x_l \leq x_p$ and hence $\phi_{1,t}'(x_l) \geq \phi_{1,t}'(x_p)$.  

On the other hand, for $l \geq p+1$, we have $x_l \geq x_{p+1}$. Hence, for $x \in [x_p, x_{p+1}]$, the difference 
\begin{align}
	e_l(x) - e_{p+1}(x) 
	&= 		a_l (x - x_l) + \alpha_t \cdot \phi_{1,t}(x_l) - \alpha_t \cdot \phi_{1,t}(x) \notag \\
	&\qquad - [ a_{p+1} (x - x_{p+1}) + \alpha_t \cdot \phi_{1,t}(x_{p+1}) - \alpha_t \cdot \phi_{1,t}(x)  ] \notag \\
	& = \alpha_t [\phi_{1,t}'(x_l)(x - x_l) -   \phi_{1,t}'(x_{p+1})(x - x_{p+1}) +   \phi_{1,t}(x_l)  -   \phi_{1,t}(x_{p+1}) ].  \notag 
\end{align}
Note that $\phi_{1,t}$ is strictly concave, $-\phi_{1,t}$ is strictly convex. Hence, it has a first-order lower bound
$
-\phi_{1,t}(x_{p+1}) \geq -\phi_{1,t}(x_l)  - \phi_{1,t}'(x_l)(x_{p+1} - x_l)
$
which leads to
\begin{align}
	e_l(x) - e_p(x) 
	& \geq	 \alpha_t [\phi_{1,t}'(x_l)(x - x_l) -   \phi_{1,t}'(x_{p+1})(x - x_{p+1})  - \phi_{1,t}'(x_l)(x_{p+1} - x_l) ]  \notag \\
	& =\alpha_t \underbrace{[\phi_{1,t}'(x_l)-   \phi_{1,t}'(x_{p+1})   ]}_{\leq 0} \underbrace{(x - x_{p+1})}_{\leq 0} \notag \\
	& \geq 0 \label{ineq: e_l - e_{p+1}}
\end{align}
where the last inequality holds since $\phi_{1,t}$ is strictly concave and increasing, it implies that $\phi_{1,t}'>0$ and $\phi_{1,t}'$ is strictly decreasing; that is, for $l \geq p+1$, it follows that $x_l \geq x_{p+1}$ and hence $\phi_{1,t}'(x_l) \leq \phi_{1,t}'(x_{p+1})$.  

With Inequality~\eqref{ineq: e_p leq e_{p+1}} and ~\eqref{ineq: e_l - e_{p+1}}, we have that for $x\in [x_p,x^{\prime}] \subseteq [x_p, x_{p+1}]$ and $l \geq p+1$,
$
e_l(x) \geq e_{p+1}(x) \geq e_p(x),
$
and in combination with Inequalities~\eqref{ineq: e_l - e_p} that $e_l(x) \geq e_p(x)$ for $l \leq p$ and $x \in [x_p, x^{\prime}] \subseteq [x_p, x_{p+1}]$. Therefore, we obtain 
$
e_l(x) \geq e_p(x)
$
for all $l$ and $x \in [x_p, x^{\prime}]$.

In addition, we now show that
$
e_l(x) \geq e_{p+1}(x)
$ for all $l$ and $x \in (x^{\prime}, x_{p+1}]$.
With Inequalities~\eqref{ineq: e_p > e_{p+1}} and~\eqref{ineq: e_l - e_p}, it follows that for  $x \in (x^{\prime}, x_{p+1}] \subseteq [x_p, x_{p+1}]$ and $l \leq p$,
we have $e_l(x) \geq e_p(x) > e_{p+1}(x)$, and by Inequalities~\eqref{ineq: e_l - e_{p+1}}, we have $e_l(x) \geq e_{p+1}(x)$ for all $x \in (x^{\prime}, x_{p+1}]  \subseteq [x_p, x_{p+1}]$ and $l \geq p+1$. Therefore, we obtain
$
e_l(x) \geq e_{p+1}(x)
$
for all $l$ and $x \in (x^{\prime}, x_{p+1}]$.
Hence,
\begin{align} \label{eq: mini error formula for x}
	\min_{l = 0,\dots,M_x} e_l(x) =
	\begin{cases}
		e_p(x), & \text{if } x \in [x_p, x^{\prime}] \\
		e_{p+1}(x), & \text{if } x \in (x^{\prime}, x_{p+1}].
	\end{cases}
\end{align}
Hence, using Equality~\eqref{eq: mini error formula for x}, we obtain 
\begin{align} \label{eq: sup error for x}
	\sup_{x \in [x_p, x_{p+1}]} \min_{l = 0,1,\dots,M_x} e_l(x) =
	\begin{cases}
		\sup_{x \in [x_p, x_{p+1}]} e_p(x), & \text{if } x \in [x_p, x^{\prime}]\\
		\sup_{x \in [x_p, x_{p+1}]} e_{p+1}(x), & \text{if } x \in (x^{\prime}, x_{p+1}].
	\end{cases}
\end{align}

Moreover, with the aids of monotonicity,  $e_p(x^{\prime}) \geq e_p(x)$ for all $x\in[x_p, x^{\prime}]$ and $e_{p+1}(x^{\prime}) \geq e_{p+1}(x)$ for all $x \in [x^{\prime}, x_{p+1}]$, Equality~\eqref{eq: sup error for x} becomes
\begin{align*} 
	\sup_{x \in [x_p, x_{p+1}]} \min_{l = 0,1,\dots,M_x} e_l(x) =
	\begin{cases}
		e_p(x^{\prime}), & \text{if } x \in [x_p, x^{\prime}]\\
		e_{p+1}(x^{\prime}), & \text{if } x \in [x^{\prime}, x_{p+1}].
	\end{cases}
\end{align*}
Then by the fact that $e_p(x^{\prime})=e_{p+1}(x^{\prime})$, we obtain 
$$
\sup_{x \in [x_p, x_{p+1}]} \min_{l = 0,1,\dots,M_x} e_l(x) = e_p(x^{\prime})=e_{p+1}(x^{\prime}).
$$
Solving the equation $e_p(x^{\prime}) = e_{p+1}(x^{\prime})$ yields
$$
x^{\prime}  = x^{\prime }(x_{p+1})   = \frac{	 \phi_{1,t}'(x_p)   x_p - \phi_{1,t}'(x_{p+1}) x_{p+1}  +   \phi_{1,t} (x_{p+1})  -  \phi_{1,t} (x_p)  }{\phi_{1,t}'(x_p) -  \phi_{1,t}'(x_{p+1}) }.
$$

An almost identical proof would work for  analyzing the approximation error for $c$, hence we omitted.  \qedhere

\end{proof}

\begin{proof}[Proof of Theorem~\ref{theorem: Successive Partition Points for general Additively Separable Utility}]
Fix $\varepsilon > 0$, $x_p$ and $c_q$ for $p = 0,1,\dots, M_x-1$ and $q = 0, 1, \dots, M_C -1$. We take the additively separable utility $f(x, c) = U_t((1+x)(1-c))$. By Lemma~\ref{lemma: Maximum Approximation Errors}, we choose $\varepsilon_x > 0$ and $\varepsilon_c > 0$ such that $e_p( x^{\prime} ) \leq \varepsilon_x$ and $e_q( c^{\prime}) \leq \varepsilon_c$, where $\varepsilon = \varepsilon_x + \varepsilon_c$. 
Subsequently, with the given partition points $x_p$ and $c_q$ that build the hyperplane $h_{p, q}(x, c) := \{(x,c) : a_p x + b_qc + \gamma_{p, q} = 0\}$, we now construct the next two hyperplanes: $h_{p+1, q}$ and $h_{p, q+1}$.
Specifically, fix $c = c_q$. Note that~$e_q(c_q) =0$, we observe that 
\begin{align*}
	\sup_{ \substack{ x \in [x_p, x_{p+1}]\\ c \in [c_q, c_{q+1}]} } e(x, c) 
	= \sup_{x \in [x_p, x_{p+1}]} \min_{l} e_l(x) + \sup_{c \in [c_q, c_{q+1}]} \min_{r} e_r(c) 
	&  = e_p(x^{\prime}) + e_q(c_q)\\ 
	&  = e_p(x^{\prime})   + 0. 
\end{align*}
Moreover, set
$
e_p( x^{\prime} ) = \alpha_t \phi_{1, t}'(x_p) ( x' - x_p ) + \alpha_t \phi_{1, t}(x_p) - \alpha_t \phi_{1, t}(x'):= \varepsilon_x.
$
This implies that
\begin{align*}
	\frac{\varepsilon_x }{ \alpha_t } 
	&= \phi_{1, t}'(x_p)  \left[  \frac{	 \phi_{1,t}'(x_{p+1})(x_p- x_{p+1} ) +   \phi_{1,t} (x_{p+1})  -  \phi_{1,t} (x_p)   }{\phi_{1,t}'(x_p) -  \phi_{1,t}'(x_{p+1}) }  \right] + \phi_{1, t}(x_p) - \phi_{1, t}(x').
\end{align*}
Take 
\begin{align}\label{eq: mathcal_A definition}
	\mathcal{A}(x_{p+1}, x_p) := \frac{	 \phi_{1,t}'(x_{p+1})(x_p- x_{p+1} ) +   \phi_{1,t} (x_{p+1})  -  \phi_{1,t} (x_p)   }{\phi_{1,t}'(x_p) -  \phi_{1,t}'(x_{p+1}) } = x^{\prime} - x_p
\end{align}
and note that $\mathcal{A} > 0$ since $x^{\prime} > x_p.$
Then, 
$
\phi_{1, t}'(x_p) \mathcal{A}(x_{p+1}, x_p)+ \phi_{1, t}(x_p) - \phi_{1, t}( \mathcal{A}(x_{p+1}, x_p) + x_p) - \frac{\varepsilon_x }{ \alpha_t }.
$
Hence, solving the nonlinear equation $\mathcal{G}(\mathcal{A}) = 0$ with
$
\mathcal{G} (\mathcal{A})=  \phi_{1, t}'(x_p) \mathcal{A} - \phi_{1, t}(  \mathcal{A} + x_p ) + \phi_{1, t}(x_p)  -\frac{\varepsilon_x }{ \alpha_t }
$ 
yields the corresponding solution, denoted by $\mathcal{A}^*$.  The existence and uniqueness of the solution~$\mathcal{A} > 0$ can be established as follows:
Indeed, noting that $\mathcal{A} \to 0$,  $\mathcal{A} \to -\frac{\varepsilon_t}{\alpha_t} < 0$. Moreover, note that $\phi_{1, t}$ is strictly concave, $\phi_{1, t}^\prime$ is decreasing; hence, $\mathcal{G}^{\prime}(\mathcal{A}) = \phi_{1, t}'(x_p) - \phi_{1, t}^{\prime}(  \mathcal{A} + x_p ) >0$, which shows that $\mathcal{G}$ is strictly increasing, Therefore, applying the Intermediate Value Theorem, there exists a solution $\mathcal{A}^{\prime}$ such that $\mathcal{G}(\mathcal{A}^\prime) = 0$. Moreover, the strictly increasingness of $\mathcal{G}$ assures the uniqueness of the solution.
Then, with the aid of~\eqref{eq: mathcal_A definition}, we~have
\[
\mathcal{A}^*  = \mathcal{A}^*(x_{p+1}, x_p) = \frac{	 \phi_{1,t}'(x_{p+1})(x_p- x_{p+1} ) +   \phi_{1,t} (x_{p+1})  -  \phi_{1,t} (x_p)   }{\phi_{1,t}'(x_p) -  \phi_{1,t}'(x_{p+1}) }  ,
\]
which implies that
\begin{align}
	x_{p+1}    = x_p + 	\mathcal{A}^* +  \frac{    \phi_{1,t} (x_{p+1})  -  \phi_{1,t} (x_p)  -	\mathcal{A}^* \phi_{1,t}'(x_p)  }{ \phi_{1,t}'(x_{p+1}) }.  \label{eq: x_p_1 nonlinear}
\end{align}

Take $\mathcal{B} := \mathcal{B}(x_{p+1}, x_{p}) =  \frac{    \phi_{1,t} (x_{p+1})  -  \phi_{1,t} (x_p)  -	\mathcal{A}^* \phi_{1,t}'(x_p)  }{ \phi_{1,t}'(x_{p+1}) }$.
Then we obtain $x_{p+1} = x_p + \mathcal{A}^* + \mathcal{B}$. Substituting this  back into~\eqref{eq: x_p_1 nonlinear} yields another nonlinear equation $ \mathcal{H}(\mathcal{B})  = 0$ with
$$
\mathcal{H}(\mathcal{B}) := \mathcal{B} \cdot \phi_{1,t}'(x_p + \mathcal{A}^* + \mathcal{B}) -   \phi_{1,t} (x_p + \mathcal{A}^* + \mathcal{B}) +  \phi_{1,t} (x_p) +	\mathcal{A}^* \phi_{1,t}'(x_p). 
$$
A similar argument using the Intermediate Value Theorem and strict monotonicity assures that there exists a unique solution, call it $\mathcal{B}^*$ such that $\mathcal{H}(\mathcal{B}) = 0$. Therefore, we obtain the final recursive equation:
$
x_{p+1} = x_p + \mathcal{A}^* + \mathcal{B}^*.
$

An almost identical proof would work for showing the similar recursive expression is valid for $c$. Hence, we omitted.
\end{proof}

\begin{proof}[Proof of Theorem~\ref{theorem: Successive Partition}]
Fix $\varepsilon > 0$, $x_p$ and $c_q$ for $p = 0,1,\dots, M_x-1$ and $q = 0, 1, \dots, M_C -1$. We take the log-additive separable utility $f(x, c) = U_t((1+x)(1-c)) = \log(1+x) + \log(1-c)$. By Lemma~\ref{lemma: Maximum Approximation Errors}, we choose $\varepsilon_x > 0$ and $\varepsilon_c > 0$ such that $e_p( x^{\prime} )\leq \varepsilon_x$ and $e_q( c^{\prime}) \leq \varepsilon_c$, where $\varepsilon = \varepsilon_x + \varepsilon_c$. 
Subsequently, with the given partition points $x_p$ and $c_q$ that build the hyperplane $h_{p, q}(x, c) := \{(x,c) : a_p x + b_qc + \gamma_{p, q} = 0\}$, we now construct the next two hyperplanes: $h_{p+1, q}$ and $h_{p, q+1}$.
Specifically, fix $x=x_p$. Note that $e_p(x_p) =0$, a lengthy but straightforward calculation leads to
\begin{align*}
	& \sup_{x}  \min_l e_l(x) + \sup_{c} \min_r e_r(c) \\
	& \qquad = e_p(x_p) + e_q(c^{\prime})\\ 
	& \qquad = 0 + e_q(c^{\prime}) \\
	& \qquad = b_qc^{\prime}+\log(1-c_q)-b_qc_q-\log(1-c^{\prime})\\
	& \qquad = \frac{1-c_{q+1}}{c_{q+1} - c_q} \log \left( \frac{1-c_q}{ 1-c_{q+1} } \right) - \log \left( \frac{1 - c_{q+1}}{c_{q+1} - c_q} \log \left( \frac{1 - c_q}{1 - c_{q+1}} \right) \right) -1.
\end{align*}
Now consider an auxiliary function $f_c(\theta) := \theta - \log\theta-1-\varepsilon_c$, then 
$
\theta_c = \frac{1 - c_{q + 1 }}{c_{q+1}-c_q}\log \left( \frac{1 - c_q}{1 - c_{q+1}} \right),
$
which solves $f_c(\theta)=0$. Since $e_{q+1}(c_{q+1})$ is strictly increasing in $[c_q, c_{\max}]$, $\theta_c$ is uniquely defined. Moreover, note that 
\begin{align*}
	\theta_c 
	= \frac{1 - c_{q+1}}{c_{q+1} - c_q} \log \left( \frac{1 - c_q}{1 - c_{q+1}} \right)
	&= \frac{\frac{ 1 - c_{q+1} }{1 - c_q}}{\frac{c_{q+1} - c_q}{1 - c_q}}\log \left(\frac{\frac{1-c_q}{1 - c_q}}{\frac{1 - c_{q+1}}{1 - c_q}} \right)\\
	&= \frac{1- \mathsf{d} }{\mathsf{d}}\log \left( \frac{1}{1- \mathsf{d}} \right)
\end{align*}
where 
$
\mathsf{d}  = \frac{ c_{q+1} - c_q}{ 1 - c_q} := \mathsf{d}_c.
$
Then, it follows that $c_{q+1} = (1 - \mathsf{d}_c) c_q + \mathsf{d}_c$; hence, the successive hyperplane $h_{p, q+1}$ is found.

Similarly, we now prove the successive recursion for $x$. Fix $c = c_q$. Note that $e_q(c_q) =0$, a lengthy but straightforward calculation leads to
\begin{align*}
	&\sup_{x} \min_l e_l(x) + \sup_{c} \min_r e_r(c) \\
	& \qquad = e_p(x^{\prime}) + e_q(c_q)\\ 
	& \qquad = e_p(x^{\prime}) + 0 \\
	& \qquad = a_p x^{\prime}+\log(1 + x_p) - a_px_p-\log(1 + x^{\prime}) \\
	& \qquad = \frac{1+x_{p+1}}{x_{p+1}-x_p}\log \left( \frac{1+x_{p+1}}{1+x_p} \right)-\log \left( \frac{1+x_{p+1}}{x_{p+1}-x_p}\log(\frac{1+x_{p+1}}{1+x_p}) \right) - 1.
\end{align*}
Then consider another auxiliary function $f_x( \mathsf{b} ):= \mathsf{b} - \log\mathsf{b} - 1 - \varepsilon_x$. Then
$
\mathsf{b}_x:=\frac{1 + x_{p+1}}{ x_{p+1} - x_p} \log \left( \frac{1+x_{p+1}}{1+x_p} \right),
$
which solves $f_x(\mathsf{b}) = 0$. Since $e_{p+1}(x_{p+1})$ is strictly increasing in $[x_p, x_{\max}]$, $\beta_x$ is uniquely defined. Moreover, note~that 
\begin{align*}
	\mathsf{b}_x 
	= \frac{1 + x_{p+1}}{x_{p+1} - x_p} \log \left( \frac{1 + x_{p+1}}{1 + x_p} \right) 
	&= \frac{\frac{1 + x_{p+1}}{1 + x_p}}{\frac{x_{p+1} - x_p}{1+x_p}} \log \left( \frac{1 + x_{p+1}}{1+x_p} \right)\\
	&= \frac{1+ \mathsf{a}}{\mathsf{a}}\log( 1 + \mathsf{a})
\end{align*}
where
$
\mathsf{a} = \frac{x_{p+1} - x_p}{1+x_p} := \mathsf{a}_x > 0.
$
Hence, it follows that $x_{p+1} = (1 + \mathsf{a}_x) x_p + \mathsf{a}_x$, and the successive hyperplane $h_{p+1, q}$ is built.
\end{proof}

\begin{proof}[Proof of Lemma~\ref{lemma: Optimal Number Of Hyperplanes}]
According to Lemma~\ref{theorem: Separable Maximum Approximation Error}, partition points $\{ x_l \}_{l \geq 0}$ and $\{ c_r \}_{r \geq 0}$ of hyperplanes are determined separately by given the associated error tolerances $\varepsilon_x$ and $\varepsilon_c$. Hence, the optimal number of hyperplanes required to meet the approximation error $\varepsilon  = \varepsilon_x + \varepsilon_c$ is $M:= M_x + M_c$.
\end{proof}

\section{Some Technical Results} \label{appendix: some technical results}

This appendix collects technical results related to the running expected objective function:
$
J_p(t; K(t), K(t-1)) := \mathbb{E}_p \left[ U_t \left(  \frac{V(t) }{V(t-1)} \right) \right].
$

\renewcommand{\thelemma}{\Alph{section}.\arabic{lemma}}

\begin{lemma}\rm \label{lemma: data-driven expression of the running ELG}
The running expected logarithmic growth of the portfolio satisfies
\begin{align*}
	J_p(t; K(t), K(t-1)) 
	&= \sum_{j=1}^m p_j \left[ U_t \left(   (1 + K(t)^\top  x^j)(1 - |K(t) - K(t-1)|^\top C(t)) \right) \right].
\end{align*}
\end{lemma}

\begin{proof}
By Equality~\eqref{eq: running expected utility}, we obtain that
{\small	\begin{align*}
	J_p(t; K(t), K(t-1)) 
	&= \mathbb{E}_p \left[ U_t \left(  \frac{V(t) }{V(t-1)} \right) \right]\\
	& = \int_{\mathbb{R}^n} U_t \left(   (1 + K(t)^\top  x)(1 - |K(t) - K(t-1)|^\top C(t)) \right) f_X(x)dx,
\end{align*}
}	where $f_X(x) = \sum_{ j=1 }^m p_j \delta(x - x^j)$ is the probability distrubtion with Dirac Delta functions representing the random return $X$ taking values~$x^j$ with probabilities $p^j$. Hence, it follows that
	\begin{align*}
&	J_p(t; K(t), K(t-1)) \\
	& \qquad = \int_{\mathbb{R}^n} U_t \left(   (1 + K(t)^\top  X(t))(1 - |K(t) - K(t-1)|^\top C(t)) \right) \sum_{ j=1 }^m p_j \delta(x-x^j) dx \\
	&\qquad = \sum_{j=1}^m p_j \left[ U_t \left(   (1 + K(t)^\top  x^j)(1 - |K(t) - K(t-1)|^\top C(t)) \right) \right],
\end{align*}
which completes the proof.
\end{proof}

The next result indicates that the running expected objective is jointly concave, e.g., see \cite{bekjan2004joint}. 

\begin{lemma}[Joint Concavity of ELG] \label{lemma: joint concavity of ELG objective function}
Let $U_t$ be a additively separable utility satisfying Definition~\ref{definition: additively separable utility}.
Then the running objective $J_p(t; K(t), K(t-1))  =\mathbb{E}_p \left[ U_t \left(  \frac{V(t+1)}{V(t)} \right) \right]$ is jointly concave in $K(t)$ and $K(t-1)$.
\end{lemma}

\begin{proof}  
With the aid of Lemma~\ref{lemma: data-driven expression of the running ELG},  we begin by noting that
$$
\mathbb{E}_p \left[ U_t \left(  \frac{V(t+1)}{V(t)} \right) \right] =  \sum_{j=1}^m p_j \left[ U_t \left( (1 + K(t)^\top  x^j) \left( 1 - |K(t) - K(t-1)|^\top C(t) \right) \right) \right].
$$
If the inner term $U_t \left( (1 + K(t)^\top  x^j) \left( 1 - |K(t) - K(t-1)|^\top C(t) \right) \right)$ is jointly concave for $K(t)$ and $ K(t-1)$, then, with $p_j \geq 0$ and $\sum_j p_j = 1$, the objective function
$$
\sum_{j=1}^m p_j \left[ U_t \left( ( 1 + K(t)^\top  x^j ) \left( 1 - | K(t) - K(t-1) |^\top C(t) \right) \right) \right]
$$
is also concave.
To establish this, we employ the fact that $U_t$ is additively separable; i.e., there exists functions $\phi_{1, t}, \phi_{2,t}$ and constants $\alpha_t>0, \beta_t>0$ such that $\phi_{1, t}$ is strictly concave and strictly increasing, and $\phi_{2,t}$ is strictly concave and strictly decreasing, and we have
\begin{align*}
&U_t \left( (1 + K(t)^\top  x^j) \left( 1 - | K(t) - K(t-1) |^\top C(t) \right) \right) \\
& \qquad = \alpha_t \phi_{1, t}( K(t)^\top  x^j) + \beta_t \phi_{2, t} \left( |K(t) - K(t-1)|^\top C(t) \right).
\end{align*}
The first term on the right-hand side is concave in $K(t)$ since $\alpha_t >0$, $K(t)^\top x^j$ is linear, and $\phi_{1,t}$ is concave. For the second term, it suffices to show that 
$
g(K(t), K(t-1)) := |K(t)-K(t-1)|^\top C(t)
$
is jointly convex. Then, by the convex composition rule, $\beta_t \phi_{2, t}( g(K(t), K(t-1)))$ is concave. 
To see this, let $\overline{K}^1(t), \overline{K}^2(t), \overline{K}^1(t-1), \overline{K}^2(t-1) \in \mathcal{K}$ be four different vectors and $\lambda \in [0,1]$. By the triangle inequality, we have
\begin{align*}
	&g (\lambda \overline{K}^1(t) + (1-\lambda)\overline{K}^1(t-1),\; \lambda \overline{K}^2(t) + (1-\lambda) \overline{K}^2(t-1))\\
	& \qquad = |\lambda\overline{K}^1(t) - \lambda \overline{K}^1(t-1) + (1-\lambda)\overline{K}^2(t) - (1-\lambda)\overline{K}^2(t-1)|^\top C(t)  \\
	&\qquad \leq \lambda| \overline{K}^1(t) - \overline{K}^1(t-1)|^\top C(t) + (1-\lambda)|\overline{K}^2(t) - \overline{K}^2(t-1)|^\top C(t)\\
	&\qquad = \lambda g (\overline{K}^1(t), \overline{K}^1(t-1)) + (1-\lambda) g(\overline{K}^1(t), \overline{K}^1(t-1)).
\end{align*}
Hence, $g(K(t), K(t-1)) :=|K(t)-K(t-1)|^\top C(t)$ is jointly convex. By the convex composition rule, see \cite[Section 3.2.4]{boyd2004convex}, it follows that $\beta_t \phi_{2, t}(|K(t)-K(t-1)|^\top C(t) )$ is concave. Therefore, $U_t$ is jointly concave in $K(t)$ and $K(t-1).$
\end{proof}

\begin{corollary}
Let $t = 1,2,\dots,T$, given $K(t-1)$, the running expected objective maximization problem
$
	\max_{K(t) \in \mathcal{K}} J_p(t; K(t), K(t-1) ) 
$
is a convex program with a concave objective function.
\end{corollary}

\begin{proof}
By Lemma~\ref{lemma: joint concavity of ELG objective function}, $J_p(t; K(t), K(t-1) )$ is jointly concave in $K(t)$ and $K(t-1)$. Moreover, since $\mathcal{K}$ is a convex compact set, the problem stated above is a convex program with a concave objective function.
\end{proof}

\end{document}